\documentclass[sigconf, nonacm]{acmart}

\usepackage{color}  
\usepackage{hyperref}
\usepackage{balance}
\usepackage{graphicx}
\usepackage{subfigure}
\usepackage[graphicx]{realboxes}
\usepackage{balance}
\usepackage[ruled, vlined, linesnumbered]{algorithm2e}
\usepackage{array}
\usepackage{multirow}
\usepackage[normalem]{ulem}
\useunder{\uline}{\ul}{}
\usepackage{booktabs}
\usepackage{enumitem}
\usepackage{amsmath}
\usepackage{amsthm} 

\newtheorem{theorem}{Theorem}
\newtheorem{example}{Example}
\newtheorem{definition}{Definition}
\newtheorem{proposition}{Proposition}

\newcommand\vldbdoi{XX.XX/XXX.XX}
\newcommand\vldbpages{XXX-XXX}
\newcommand\vldbvolume{16}
\newcommand\vldbissue{1}
\newcommand\vldbyear{2022}
\newcommand\vldbauthors{\authors}
\newcommand\vldbtitle{\shorttitle} 
\newcommand\vldbavailabilityurl{https://github.com/Yocenly/TNCP}
\newcommand\vldbpagestyle{plain} 

\begin{document}
\title{Targeted \texorpdfstring{$k$}--node Collapse Problem: Towards Understanding the Robustness of Local \texorpdfstring{$k$}--core Structure}

\author{Yuqian Lv}
\affiliation{%
  \institution{Zhejiang University of Technology}
}
\email{lvyuqian_email@163.com}

\author{Bo Zhou}
\affiliation{%
  \institution{Zhejiang University of Technology}
}
\email{wxjs201@163.com}

\author{Jinhuan Wang}
\affiliation{%
  \institution{Zhejiang University of Technology}
}
\email{jhwang@zjut.edu.cn}

\author{Qi Xuan}
\affiliation{
  \institution{Zhejiang University of Technology}
}
\email{xuanqi@zjut.edu.cn}





\begin{abstract}
  The concept of \texorpdfstring{$k$}--core, which indicates the largest induced subgraph where each node has \texorpdfstring{$k$}- or more neighbors, plays a significant role in measuring the cohesiveness and the engagement of a network, and it is exploited in diverse applications, e.g., network analysis, anomaly detection, community detection, etc. Recent works have demonstrated the vulnerability of \texorpdfstring{$k$}--core under malicious perturbations which focuses on removing the minimal number of edges to make a whole \texorpdfstring{$k$}--core structure collapse. However, to the best of our knowledge, there is no existing research concentrating on how many edges should be removed at least to make an arbitrary node in \texorpdfstring{$k$}--core collapse. Therefore, in this paper, we make the first attempt to study the Targeted \texorpdfstring{$k$}--node Collapse Problem (TNCP) with four novel contributions. Firstly, we offer the general definition of TNCP problem with the proof of its NP-hardness. Secondly, in order to address the TNCP problem, we propose a heuristic algorithm named TNC and its improved version named ATNC for implementations on large-scale networks. After that, the experiments on \texorpdfstring{$16$}- real-world networks across various domains verify the superiority of our proposed algorithms over $4$ baseline methods along with detailed comparisons and analyses. Finally, the significance of TNCP problem for precisely evaluating the resilience of \texorpdfstring{$k$}--core structures in networks is validated.
\end{abstract}

\maketitle

\pagestyle{\vldbpagestyle}
\begingroup\small\noindent\raggedright\textbf{PVLDB Reference Format:}\\
\vldbauthors. \vldbtitle. PVLDB, \vldbvolume(\vldbissue): \vldbpages, \vldbyear.\\
\href{https://doi.org/\vldbdoi}{doi:\vldbdoi}
\endgroup
\begingroup
\renewcommand\thefootnote{}\footnote{\noindent
This work is licensed under the Creative Commons BY-NC-ND 4.0 International License. Visit \url{https://creativecommons.org/licenses/by-nc-nd/4.0/} to view a copy of this license. For any use beyond those covered by this license, obtain permission by emailing \href{mailto:info@vldb.org}{info@vldb.org}. Copyright is held by the owner/author(s). Publication rights licensed to the VLDB Endowment. \\
\raggedright Proceedings of the VLDB Endowment, Vol. \vldbvolume, No. \vldbissue\ %
ISSN 2150-8097. \\
\href{https://doi.org/\vldbdoi}{doi:\vldbdoi} \\
}\addtocounter{footnote}{-1}\endgroup

\ifdefempty{\vldbavailabilityurl}{}{
\vspace{.3cm}
\begingroup\small\noindent\raggedright\textbf{PVLDB Artifact Availability:}\\
The source code, data, and/or other artifacts have been made available at \url{\vldbavailabilityurl}.
\endgroup
}

\section{Introduction}

\begin{figure}[t]
  \centering
	\includegraphics[width=0.8\linewidth]{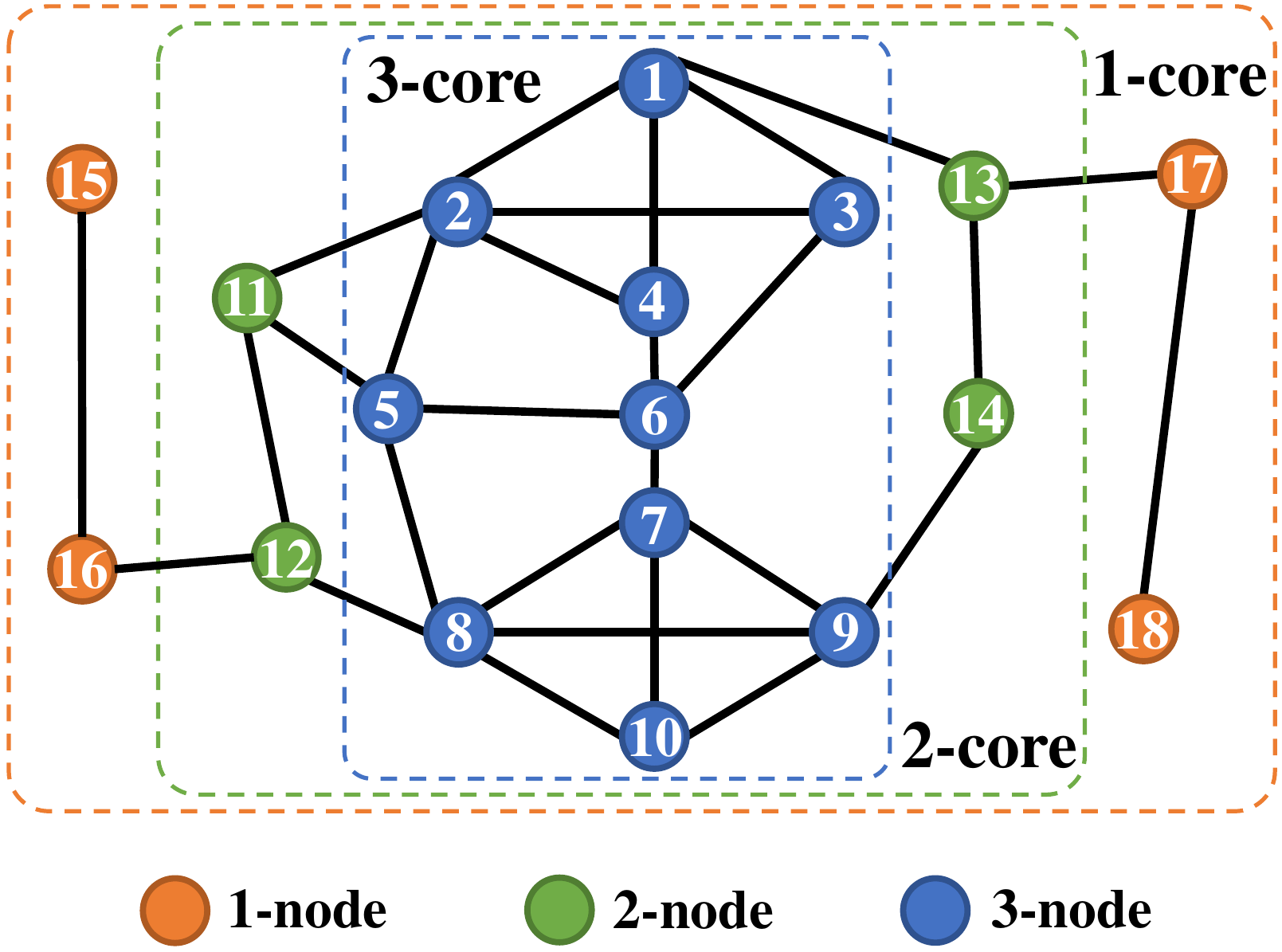}
  \caption{An example of $k$-core distribution in a graph. Each subgraph in the dotted box with a certain color represents the $k$-core and the color of each node represents its core value.}
  \label{fig: k_core}
\end{figure}

Networks or graphs play significant roles in describing various complex systems from numerous domains, e.g., social networks\cite{majeed2020graph,wang2015community,liu2017experimental,girvan2002community}, citation networks\cite{liu2019link,hummon1989connectivity}, biological networks\cite{koutrouli2020guide,yu2013review,girvan2002community} and power networks\cite{pasqualetti2011graph,akinyele2014review}. Therefore, understanding the topological information of graphs is what matters in the study of graph theory. Due to the advantages of simplicity and efficiency \cite{kong2019k}, the concept of $k$-core, which denotes the maximal induced subgraph where each node within it occupies at least $k$ neighbors \cite{dorogovtsev2006k}, has stood out as an important metric for describing the global structural engagement of networks from massive evaluation metrics. As shown in Figure \ref{fig: k_core}, an example graph with $18$ nodes and $28$ edges is given where $3$ cores exist, i.e., $1$-core, $2$-core and $3$-core which are surrounded by dotted boxes with different colors. As more and more researchers devoted themselves to the study of $k$-core, $k$-core has been used in a broad variety of important applications \cite{malliaros2020core}. For example, in ecological networks, Morone et al. \cite{morone2019k} exploited the $k$-core as a predictor to estimate the structural collapse in mutualistic ecosystems, and Burleson-Lesser et al. \cite{burleson2020k} presented a new approach for characterizing the stability and robustness of networks with all-positive interactions by studying the distribution of the $k$-core of the underlying network. Besides, in social networks, Wang et al. \cite{wang2021vulnerability} considered the pruning process of $k$-core to measure the vulnerability and resilience of social engagement and further studied its equilibrium statistical mechanics. And in biological networks, Luo et al. \cite{luo2009core} studied the core structure of protein-protein interactions networks with an interesting discovery that core structures help to reveal the existence of multiple levels of protein expression dynamics, and Isaac et al. \cite{isaac2015analysis} discovered that residues belonging to inner cores are more conserved than those at the periphery of the network with the evidence that these groups are functionally and structurally critical.

With the rapidly increasing number of applications based on $k$-core structures, the robustness (or also called resilience) of $k$-core have gradually attracted the attention of researchers. For instance, Zhou et al.\cite{zhou2021robustness} studied the robustness of $k$-shell, a subset of $k$-core, and demonstrated that $k$-shell is vulnerable under the disturbance of edge rewiring. Their optimal-based experimental results showed that the $k$-core distributions of graphs can be drastically changed even a small proportion of edges are rewired. Zhou et al.\cite{zhou2022attacking} also studied the minimal budgets of removed edges for the collapse of the innermost $k$-core. They provided a proof of its NP-hardness and offered effective heuristic algorithms to cover this problem. Furthermore, Chen et al.\cite{chen2021edge} focused on the $k$-core minimization problem and suggested three sub-problems, i.e., KNM, KEM and KCM. They further proposed several heuristic algorithms through edge removal to cover these sub-problems respectively. Medya et al. \cite{ijcai2020p480} also concentrated on the $k$-core minimization problem and proposed a novel algorithm inspired by shapley value, a cooperative game-theoretic concept. Their algorithm could leverage the strong interdependencies in the effects of edges removal in the search space. Besides, Zhang et al.\cite{zhang2017finding} studied the collapsed $k$-core problem which aims to find a set of nodes whose detachment will lead to the minimal size of the resulting collapsed $k$-core.

However, as we can see, throughout the previous works, all of them considered the $k$-core as a whole to evaluate its robustness, while none of them focused on the robustness of an individual node within $k$-core. Thus it brings us a question that how many edges should we disconnect at least to make an arbitrary node contained in $k$-core collapse? As far as we know, there is no existing work dedicating to the study of this problem. In this paper, we make the first attempt to study this problem and name it as Targeted $k$-node Collapse Problem (TNCP). Our main contributions can be summarized as below.

\begin{itemize}[itemsep= 0 pt, topsep = 0 pt, itemindent = 0 pt, leftmargin = 15 pt]
  \item We offer a general definition of TNCP problem with a proof of its NP-hardness. We demonstrate that the naive exhaustive method will lead to the exponential explosion of the time complexity. Therefore, a series of theorems are provided to narrow down the search space of candidates.
  \item Combined with the theorems, we propose a heuristic algorithm named TNC to cover the TNCP problem. However, we find that TNC algorithm is not suitable for large-scale networks. Thus, an improved algorithm named ATNC with less time complexity is proposed based on TNC algorithm.
  \item We verify the superiority of our proposed algorithms over $4$ baseline methods through experiments on $16$ real-world networks collected from different public platforms along with detailed comparisons and analyses.
  \item We demonstrate that the research of TNCP problem is helpful for precisely evaluating the resilience of the $k$-core structures in networks.
\end{itemize}

The remaining sections of this paper are structured as follows. In Section \ref{sec: relatedworks}, a brief review on the previous works about $k$-core is illustrated. In Section \ref{sec: problemstatement}, the statement of TNCP problem and basic definitions, which will be used in the rest of this paper, are introduced along with the theorems for candidate reduction. In Section \ref{sec: methodology}, we introduce our proposed methods TNC and ATNC with their time complexity analyses. In Section \ref{sec: experiment}, we give the introductions about the datasets being used, the baseline methods for comparisons and the metrics for evaluations. In Section \ref{sec: results}, experimental results on all mentioned datasets are shown along with detailed comparisons and analyses between our proposed algorithms and the baseline methods. In Section \ref{sec: application}, the significance of TNCP problem for precisely evaluating the resilience of $k$-core in networks is validated. Finally, our work is concluded in Section \ref{sec: conclusion}.

\section{Related Works}
\label{sec: relatedworks}
The researches on the $k$-core structure of networks have been enduring, and those most related to our work are introduced as below, including core decomposition, core robustness/resilience, and core percolation.

\textbf{Core Decomposition.} Hajnal et al. \cite{hajnal1966chromatic} gave the first $k$-core related concept and defined the degeneracy of a graph as the maximum core number of a node. Then, Seidman \cite{seidman1983network}, as well as Matula and Beck \cite{matula1983smallest}, defined the $k$-core subgraph as the maximal connected subgraph where each node has at least $k$ neighbors. Khaouid et al. \cite{khaouid2015k} explored whether $k$-core decomposition of large networks can be computed using a consumer-grade PC. Sariy{\"u}ce et al. \cite{sariyuce2013streaming} proposed the first incremental $k$-core decomposition algorithms for streaming graph data. H{\'e}bert-Dufresne et al. \cite{hebert2016multi} proposed onion decomposition which is derived from $k$-core decomposition. Eidsaa and Almaas \cite{eidsaa2013s} presented $s$-core analysis, a generalization of $k$-core analysis, for weighted networks.

\textbf{Core Robustness/Resilience.} In addition to works mentioned in the last section, Adiga and Vullikanti \cite{adiga2013robust} examined the robustness of the top core sets in perturbed/sampled graphs. Zdeborov{\'a} et al. \cite{zdeborova2016fast} used $k$-core as a heuristic tool in the process of graph decycling and dismantling. Laishram et al. \cite{laishram2018measuring} proposed metrics for measuring the core resilience of a network under the situations of node/edge removals.

\textbf{Core Percolation.} Azimi-Tafreshi et al. \cite{azimi2014k} generalized the theory of $k$-core percolation on complex networks to \textbf{k}-core percolation on multiplex networks, where $\textbf{k} = (k_a,k_b,...)$. Whi et al. \cite{whi2022characteristic} revealed the hierarchical structure of functional connectivity on resting-state fMRI (rsfMRI) through the method of $k$-core percolation. Wang et al. \cite{wang2022generalized} proposed a generalized $k$-core percolation model to investigate the robustness of the higher-order dependent networks. Zheng et al. \cite{zheng2021k} studied the robustness of multiplex networks with interdependent and interconnected links under $k$-core percolation. Guo et al. \cite{guo2021percolation} applied $k$-core percolation analysis on brain structural network, suggesting that the brain networks are mostly reliable against random or $k$-core-based percolation with their structure design.

\section{Problem Statement}
\label{sec: problemstatement}
In this section, the descriptions of commonly used definitions and fundamental concepts will be discussed in the following contents along with the statement of TNCP problem and the proofs of our proposed theorems.

\subsection{Preliminaries}
\label{sec: preliminaries}
In this paper, a network or a graph (these two concepts will be used indiscriminately) is indicated as $G=(V,E)$, where $V$ and $E\subseteq (V\times V)$ represent the sets of nodes and edges respectively, which are extracted from real-world entities and the relationships between any pair of entities. As a prerequisite, we only focus on those unweighted and undirected graphs without self-loops or isolated nodes. Here, we present some fundamental definitions and related concepts which are relevant to the subsequent discussions. In Table \ref{tab: notations}, we compile a list of principal symbols and notations for convenient query. 

\begin{table}[htbp]
	\caption{Summary of notations.}
	\label{tab: notations}
	\begin{tabular*}{\hsize}{@{}@{\extracolsep{\fill}}lr@{}}
		\toprule[0.5mm]
		Notation    &Definition\\
		\midrule
      $G_k$  & the $k$-core subgraph of $G$     \\
      $d_{(i,G_k)}$ & the degree of node $i$ in $G_k$\\
      $C_{(i,G)}$ & the core value of node $i$ \\
      $SN_{(i,k,G)}$ & the supportive neighbors of the node $i$ in $G_k$\\
      $SN_{(i,G)}$ & the simplification of $SN_{(i,C_{(i,G)},G)}$\\
      $\mathcal{N}_{(i,G_k)}$ & the one-hop neighbors of node $i$ in $G_k$\\
      $CS_{(i,G)}$ & core strength of node $i$\\
      $NR_{(i,G)}$ & node robustness of node $i$\\
      $P_{(i,G)}$ & the corona pedigree of node $i$\\
      $E^P_{(i,G)}$ & those edges connected with nodes in $P_{(i,G)}$\\
		\bottomrule[0.5mm]
	\end{tabular*}
\end{table}

\begin{definition}
  \label{def: k-core}
  \textbf{$k$-core.} For a given graph $G$, its $k$-core, denoted as $G_k=(V_k, E_k)$ where $V_k\subseteq V$ and $E_k\subseteq E$, means the maximal induced subgraph whose nodes occupy at least $k$ neighbors within $G_k$, i.e., $\forall i \in V_k, d_{(i,G_k)}\geq k$, where $d_{(i,G_k)}$ is the degree of $i$ in $G_k$.
\end{definition}
\begin{definition}
  \label{def: corevalue}
  \textbf{Core Value of Node.} With the concept of $k$-core, we can also describe the core value of a given node $i$ within $G$ by $C_{(i,G)}$, which represents the maximum core value of the $k$-core where node $i$ exists, i.e., $C_{(i,G)}$ satisfies that $i\in G_{C_{(i,G)}}$ but $i\notin G_{C_{(i,G)}+1}$. The nodes whose core values are equal to $k$ are named as $k$-nodes. 
\end{definition}

\begin{figure}[ht]
  \centering
  \includegraphics[width=0.8\linewidth]{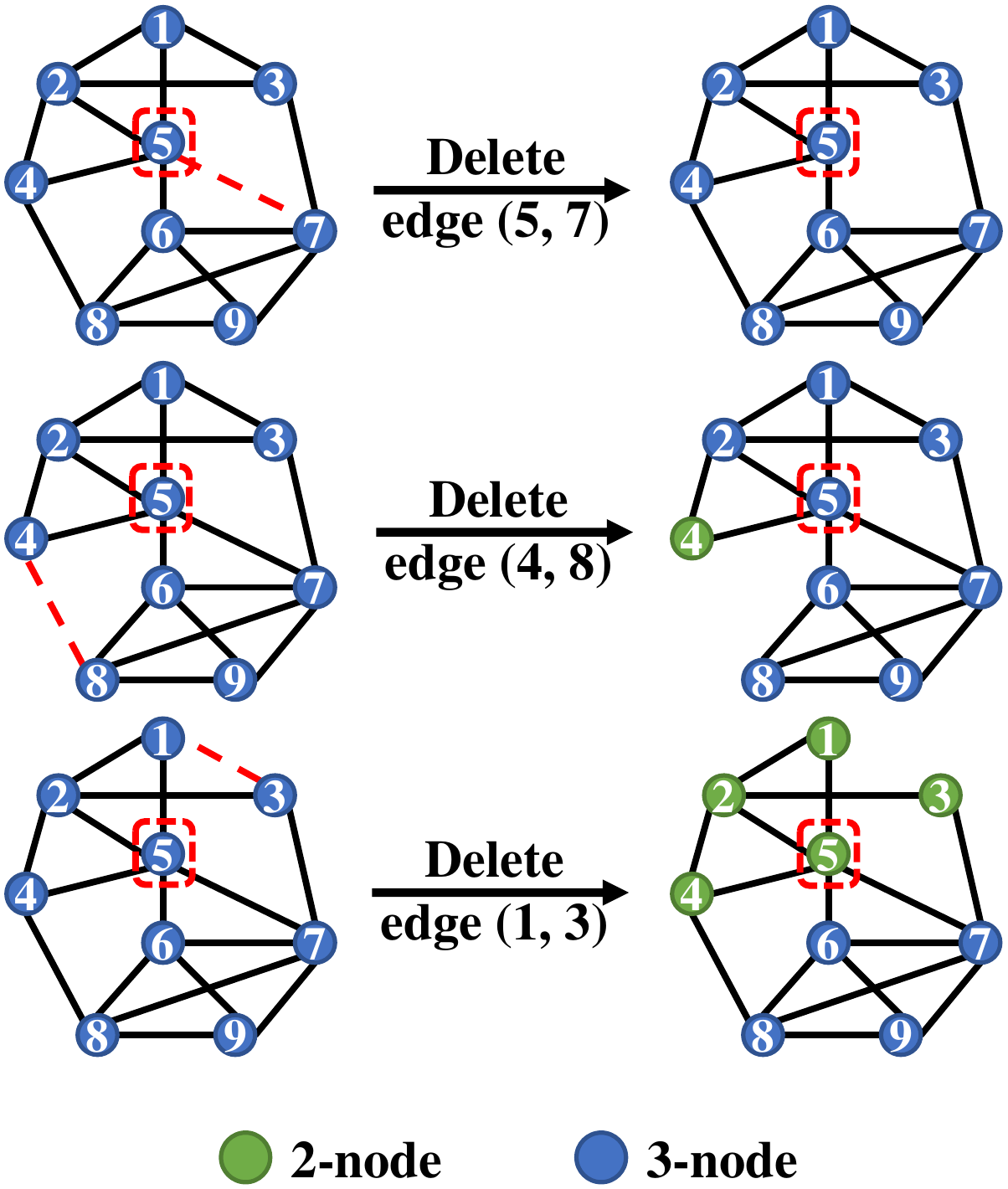}
  \caption{Given a graph with $9$ nodes and $17$ edges, we can find that all nodes stay in $3$-core. Take node $5$ as our target node, (i) the removal of edge $(5, 7)$ does not make any effect to the $k$-core distribution; (ii) the removal of edge $(4, 8)$ makes node $4$ being squeezed out of $3$-core while makes no effect to node $5$; (iii) the removal of edge $(1, 3)$ makes the target node $5$ being squeezed out of $3$-core.}
  \label{fig: deletion}
\end{figure}

In accordance with Definition \ref{def: k-core}, the existence of a given node $i$ within $G_k$ relies on its neighbor nodes who overlap with $G_k$. We can also realize that those neighbors with core values less than $k$ are not included in $G_k$. Be a result, those neighbors helping support the existence of $i$ in $G_k$ are referred to as Supportive Neighbors of $i$ which is recorded as $SN_{(i,k,G)}=\{j|j\in \mathcal{N}_{(i,G)}, C_{(j,G)}\geq k\}$, where $\mathcal{N}_{(i,G)}$ represents the one-hop neighbors of $i$ within $G$. In this way, the following theorem could be deduced.

\begin{theorem}
  \label{the: coresupport}
  \textbf{Core Support Condition.} Node $i$ can remain in $G_k$ if and only if it satisfies $|SN_{(i,k,G)}|\geq k$; otherwise, it will be squeezed out of $G_k$. 
\end{theorem}
\begin{proof}
  According to the definition of supportive neighbors of node $i$, $SN_{(i,k,G)}$ actually denotes the intersection of $\mathcal{N}_{(i,G)}$ and $V_k$. Based on Definition \ref{def: k-core}, it is clear that only the satisfaction of $d_{(i,G_k)}\geq k$ can remain the existence of $i$ in $G_k$. In this way, if node $i\in G_k$, there is $|SN_{(i,k,G)}|=|\mathcal{N}_{(i,G)}\cap V_k|=|\mathcal{N}_{(i,G_k)}|=d_{(i,G_k)}\geq k$, which shows that node $i$ could be contained in $G_k$ if and only if at least $k$ neighbors whose core values are not less than $k$ are connected with it.
\end{proof}
\begin{example}
  As illustrated in Figure \ref{fig: k_core}, for node $13$ who lives in $G_2$, it has $SN_{(13,2,G)}=2$ which allows it to satisfy Theorem \ref{the: coresupport}, while it has $SN_{(13,3,G)}=1<3$ so that it cannot exist in $G_3$.
\end{example}

Theorem \ref{the: coresupport} provides us with a sufficient and necessary condition to determine whether a certain node exists in $G_k$. Derived from this, Laishram et al.\cite{laishram2018measuring} exploited a naive and easily-computed metric called Core Strength to measure the most conservative number of disconnected neighbors of node $i$ for squeezing $i$ out of $G_{C_{(i,G)}}$, which is formulated as
\begin{equation}
  CS_{(i,G)} = |SN_{(i,C_{(i,G)},G)}| - C_{(i,G)} + 1.
\end{equation}

This metric describes that if any $CS_{(i,G)}$ of supportive neighbors are disconnected with the target node $i$, it will absolutely be in violation of Theorem \ref{the: coresupport} and be squeezed out of $G_{C_{(i,G)}}$. For instance, as shown in Figure \ref{fig: deletion}, we set node $5\in G_3$ as the target node. That is easy to find that the target node has $5$ supportive neighbors $SN_{(5,3,G)}=\{1,2,4,6,7\}$ and core strength $CS_{(5,G)}=3$. We arbitrarily select $3$ supportive neighbors to disconnect, e.g. $\{4, 6, 7\}$, then the number of its supportive neighbors will be reduced to $2$ which is against what Theorem \ref{the: coresupport} restricts. Please notice that in the rest of this paper, if $k=C_{(i,G)}$, we will use $SN_{(i,G)}$ instead of $SN_{(i,k,G)}$ for the sake of simplicity.

\subsection{Problem Definition}


As mentioned in the aforementioned contents, the core strength metric describes the most conservative number of edges we should disconnect for target-node collapse. Because of so-called cascade phenomenon or domino phenomenon of $k$-core collapse \cite{goltsev2006k}, however, this metric cannot estimate the exact number of edges that must be deleted which may be less than that quantified by core strength. As an illustration, let us turn our sights back to Figure \ref{fig: deletion}, the deletion of edge $(1,3)$ will practically make node $5$ with $CS_{(5,G)}=3$ collapse from $G_3$ to $G_2$. From here, we can derive the problem named Targeted $k$-node Collapse Problem (TNCP) aiming to quantify the minimal number of edges to remove for downgrading the core value of a target node.

\begin{proposition}
  For a given $G$ and a target node $i\in V$ with $C_{(i,G)}=k$, TNCP problem aims to find a set $e\subseteq E$ containing the least number of edges such that $C_{(i,G^\prime)} < C_{(i,G)}$, where $G^\prime = (V, E\setminus e)$, and can be formulated as:
  \begin{equation}
    \begin{aligned}
      & e^{*}=\arg \min_{e} \left |e\right |, \\& s.t. C_{(i,G^\prime)} < C_{(i,G)}.
    \end{aligned}
  \end{equation}
\end{proposition}

The minimal size of $e$ is named as Node Robustness which displays the fewest number of removed edges for the collapse of the target node under elaborate perturbations and is recorded as $NR_{(i,G)}=e^{*}$. Furthermore, those nodes whose core strengths are larger than their node robustness are referred to as Bubble Nodes which are recorded as $BN=\{i|i\in V, CS_{(i,G)}>NR_{(i,G)}\}$.

\begin{theorem}
  The TNCP problem is NP-hard for $C_{(i,G)}\geq 2$.
\end{theorem}
\begin{proof}
  First, when $C_{(i,G)}=1$, according to Definition \ref{def: k-core}, it is easy to realize that some node will always remain in $G_1$ as long as at least one neighbor is connected with it. In this way, if we want a node to collapse from $G_1$ to $G_0$, we have to disconnect all of its adjacent neighbors and make it isolated from $G$, where the cost of operations is in polynomial time.

  Then, when $C_{(i,G)}\geq 2$, considering the cascade phenomenon of $k$-core collapse, a slight disturbance is able to lead a huge variation to the target node on weakening the number of its supportive neighbors. Therefore, in such a situation, the Set Cover Problem (SCP) which has been proved to be NP-hard \cite{korte2011combinatorial} can be reduced to TNCP problem. Given a universe collection $SN_{(i,G)}$ and a set of candidates $E$ which contains all edges within $G$ under the condition of target node $i$. In order to cover the TNCP problem, we have to find out a minimal-size set of edges $e\subseteq E$ such that $|SN_{(i,G)}\setminus \Phi(e)|<C_{(i,G)}$, where $\Phi(e)$ represents those collapsed nodes whose core values will be changed after the removal of $e$ from $G$.

  Additionally, paying attention to the complexity of TNCP problem, without any prior information, we have to traverse all possible combinations of the already existing edges, whose mathematical expression can be formulated as $f=\sum_{m=1}^{\delta}\binom{|E|}{m}$, where $\delta=CS_{(i,G)}$. Based on the induction formulas of $\binom{n}{m}=\binom{n-1}{m}+\binom{n-1}{m-1}$ and $\sum_{m=0}^{M}\binom{M}{m}=2^M$, the above equation could be written as
  \begin{equation}
    \label{equ: complexity}
    \begin{split}
      f & = \mathcal{O}(|E|^{\delta-1})\binom{\delta}{0} + \mathcal{O}(|E|^{\delta-2})\binom{\delta}{1} + \dots + \binom{\delta}{\delta} \\
      & = \mathcal{O}(|E|^{\delta-1}) + \mathcal{O}(|E|^{\delta-2})\cdot 2^{1} + \dots + 2^{\delta} \\ 
      & = 2^\delta + \sum_{m=1}^{\delta}\mathcal{O}(|E|^{m-1})\cdot 2^{\delta-m}
    \end{split}
  \end{equation}

  With the complexity in the amount of the exponential increase, it is evident that traversing all combinations takes non-polynomial time. Combining the aforementioned approaches, the TNCP problem cannot be addressed in polynomial time when $C_{(i,G)}\geq 2$.
\end{proof}


\begin{example}

  As seen in Figure \ref{fig: k_core} covering $18$ nodes and $28$ edges, node $6$ is chosen to be the target node for $k$-node collapse. As mentioned before, there is just one edge, like $(1,2)$, should be removed in order to achieve the collapse of node $6$. However, without the omniscient knowledge, it is difficult to locate which edge or edges are necessarily deleted. From the descriptions above, it is naturally realized that $NR_{(6,G)}\leq CS_{(6,G)}$, thus we need to visit all $\sum_{m=1}^{2}\binom{28}{m}$ combinations to identify the key edge or edges useful for $k$-node collapse under the worst situation. Fortunately, in this scenario, the computational complexity is not high because of the previous information of $NR_{(6,G)}=1$ with the removal of edge $(4, 6)$.
\end{example}

However, the robustness of the target node will always be equal to $1$, like $NR_{(7,G)}=2$ under the removal of $(1,2)$ and $(7, 8)$ as well as $NR_{(8,G)}=2$ under the removal of $(4,8)$ and $(7, 8)$ in Figure \ref{fig: deletion}. In real-world networks, the robustness of some nodes may reach tens or even hundreds, which can probably lead to an exponential increase in time consumption. Additionally, real-world networks often contain thousands or even millions of edges, making it challenging to find a feasible solution within a reasonable amount of time. Therefore, it is important to design an effective heuristic algorithm to solve the TNCP problem.


\subsection{Candidate Reduction}
\label{sec: theorem}
As mentioned above, the naive exhaustive method for solving the TNCP problem is highly complex, making it difficult to implement in practice. In order to obtain a feasible solution within a reasonable amount of time, we need to reduce the number of candidate edges. In this section, we will introduce and prove some theorems that can be used to achieve this reduction in candidates.


\begin{theorem}
  \label{the: nodeinfluence}
  $\forall (i,j)\in E$, when $C_{(i,G)}>C_{(j,G)}$, it satisfies that $i\in SN_{(j,G)}\land j\notin SN_{(i,G)}$, and when $C_{(i,G)}=C_{(j,G)}$, it satisfies that $i\in SN_{(j,G)}\land j\in SN_{(i,G)}$.
\end{theorem}
\begin{proof}
  Based on the definition of supportive neighbors, it is evident that only those neighbors with core values greater than or equal to $C_{(i,G)}$ can be contained within the supportive neighbors of node $i$, i.e., $\{j|j\in G,C_{(j,G)}<C_{(i,G)}\}\cap SN_{(i,G)}=\emptyset$. For the same reason, considering $C_{(i,G)}=C_{(j,G)}$, there exists that $\{i,j\}\in SN_{(i,G)}\cap SN_{(j,G)}$.
\end{proof}


In other words, nodes with low core values could never establish relationships that would be supportive to nodes with high core values, while nodes with high core values establish one-way relationships that would be supportive of their connected nodes with low core values. Additionally, connected nodes with the same core value become supportive neighbors to each other. This suggests that the removal of edges bridging node pairs with different core values may only affect the side holding a low core value, while the removal of edges bridging node pairs with the same core values may affect both sides. Combining the description of Theorem \ref{the: nodeinfluence}, those relationships bridging nodes with higher core values and lower core values are named as one-way supportive relationships, and those relationships bridging nodes with the same core value are named as bidirectional supportive relationships. In this way, the neighbors who control the bidirectional supportive relationships with an arbitrary node $i$ are recorded as $\widetilde{SN}_{(i,G)}=\{j|j\in N_{(i,G)}, C_{(j,G)}=C_{(i,G)}\}$.

\begin{theorem}
  \label{lem: edgeinfluence}
  If an edge $(i,j)\in E$ is removed, for all nodes in $G$, only those with core values equal to $\min(C_{(i,G)}, C_{(j,G)})$ may have their core values changed.
\end{theorem}
\begin{proof}

  It might be assumed that $C_{(i,G)}\geq C_{(j,G)}=k_{min}$ and be marked that $G^\prime=G\setminus \{(i,j)\}$. In accordance with Theorem \ref{the: coresupport}, the removal of edge $(i,j)$ will surely make node $j$ collapse if and only if $SN_{(j,k_{min},G)}=k_{min}$. After the elimination of $(i,j)$, there exists that $SN_{(j,k_{min},G^\prime)}\leq k_{min}-1$ which is absolutely in violation with Theorem \ref{the: coresupport} and makes node $j$ excluded from $G_{k_{min}}$. In addition, Sariy{\"u}ce et al. \cite{sariyuce2013streaming} and Li et al. \cite{li2013efficient} have proved that the core value of some node can decrease at most $1$ when one of its supportive neighbors is lost. Benefiting from this, node $j$ will still remain in $G^\prime_{k_{min}-1}$ and satisfy that $SN_{(j,k_{min}-1,G^\prime)}\geq k_{min}-1$.
  According to Theorem \ref{the: nodeinfluence}, the collapse of node $j$ from $G_{k_{min}}$ to $G^\prime_{k_{min}-1}$ probably leads to the collapse of those nodes contained in $\widetilde{SN}_{(j,G_{k_{min}})}$. Following like this, based on the cascade phenomenon, it is easy to find that only nodes whose core values equal to $k_{min}$ will probably collapse from $G_{k_{min}}$ to $G^\prime_{k_{min}-1}$ in the case of eliminating edge $(i,j)$. Besides, for those nodes with core values larger than $k_{min}$, according to Theorem \ref{the: nodeinfluence}, $k_{min}$-nodes make no contributions to supporting their presence in $G_{k_{min}+1}$ so that no effect will work on them after edge $(i,j)$ is removed. Meanwhile, due to the existence of those collapsed nodes in $G^\prime_{k_{min}-1}$, on the basis of Theorem \ref{the: nodeinfluence}, they still establish supportive relationships with those nodes with core values less than $k_{min}$ whose number of supportive neighbors remains the same so that no change happens to their core values after edge $(i,j)$ is removed. 
\end{proof}

Benefiting from Theorem \ref{lem: edgeinfluence}, only the removal of edges contained in $E_{k\setminus k+1}=E_k\setminus E_{k+1}=\{(u,v)|(u,v)\in E,min(C_{(u,G)},C_{(v,G)})=k\}$ will have the probability to make the target node $i$ with $C_{(i,G)}=k$ collapse, which allows us to reduce the candidates from $E$ to $E_{k\setminus k+1}$. As illustrated in Figure \ref{fig: deletion} where only a $3$-core exists, in order to make node $5$ with $C_{(5,G)}=3$ collapse, we should take $E_{3\setminus 4}=E_3=E$ into consideration. However, we may notice that the removal of edge $(1,3)$ leads to the collapse of node $5$ while none of nodes contained in this graph collapse after the removal of edge $(5,7)$, which shows a substantial difference. Therefore, the following theorem is presented to further narrow down the search space of candidate edges.

\begin{theorem}
  \label{the: edgeinfluence}
  A given edge $(i,j)\in E$ whose elimination could make nodes within $G$ collapse requires both of the following two conditions to be satisfied: (i) $min\{CS_{(i,G)}, CS_{(j,G)}\}=1$; (ii) $(CS_{(i,G)}-CS_{(j,G)})\cdot (C_{(i,G)}-C_{(j,G)})\geq 0$.
\end{theorem}
\begin{proof}
  Firstly, the condition (i) will not be satisfied if and only if neither $CS_{(i,G)}$ nor $CS_{(j,G)}$ is equal to $1$, i.e., $CS_{(i,G)}\geq 2$ and $CS_{(j,G)}\geq 2$. In such a case, node $i$ and node $j$ satisfy that $|SN_{(i,G)}|\geq C_{(i,G)} + 1$ and $|SN_{(j,G)}|\geq C_{(j,G)} + 1$. The removal of edge $(i,j)$ will absolutely not make node $i$ or node $j$ to violate Theorem \ref{the: coresupport}.

  Next, assume that the first condition has been satisfied, it might be supposed that $CS_{(i,G)}\geq CS_{(j,G)}=1$ since edge $(i, j)$ is equivalent to edge $(j, i)$ in $G$. For the core values of node $i$ and node $j$, there are three cases to consider, i.e., $C_{(i,G)}>C_{(j,G)}$, $C_{(i,G)}<C_{(j,G)}$ and $C_{(i,G)}=C_{(j,G)}$. According to Theorem \ref{the: nodeinfluence}, the removal of edge $(i,j)$ will surely make node $j$ collapse because of the violation of Theorem \ref{the: coresupport} in the cases of $C_{(i,G)}>C_{(j,G)}$ and $C_{(i,G)}=C_{(j,G)}$ while no node will collapse in the case of $C_{i,G}<C_{(j,G)}$.
\end{proof}

Combining the findings derived by Theorem \ref{lem: edgeinfluence} and Theorem \ref{the: edgeinfluence}, for the targeted collapse mission of a given node $i$ with $C_{(i,G)}=k$, those edges existing in $E_{k\setminus k+1}$ and connecting to $V^C_{(k,G)}=\{u|u\in V, C_{(u,G)}=k\land CS_{(u,G)}=1\}$ are what we should focus on and take into candidates. Actually, nodes contained in $V^C_{(k,G)}$ are so-called corona nodes of $G_k$ \cite{baxter2015critical, zhou2022attacking,baxter2011heterogeneous}, which denotes that these nodes have exactly $k$ one-hop neighbors in $G_k$. However, the subgraph constructed by corona nodes may not be connected and will probably be divided into several disconnected components. As shown in Figure \ref{fig: k_core} where six corona nodes $\{1,3,4,5,9,10\}$ exist, the component constructed by nodes $\{1,3,4\}$ is disconnected with that constructed by nodes $\{9, 10\}$, and so does that constructed by node $\{5\}$. Therefore, for simplicity of representation, we provide the following definition to represent the corona component in which a particular corona node $i$ exists.

\begin{definition}
  \label{def: coronapedigree}
  \textbf{Corona Pedigree.} For a corona node $i\in G$ with $C_{(i,G)}=k$, the corona pedigree of $i$, denoted as $P_{(i,G)}$, represents the largest-connected subgraph containing $i$ as its component and satisfies that $\forall j\in P_{(i,G)}, C_{(j,G)}=k\land CS_{(j,G)}=1$.
\end{definition}
\begin{example}
  As shown in Figure \ref{fig: k_core}, there exist three corona pedigrees in $G_3$, e.g., $P_{(4,G)}$ contains nodes $\{1,3,4\}$ and edges $\{(1,3), (1,4)\}$, $P_{(5,G)}$ contains node $\{5\}$, $P_{(10,G)}$ contains nodes $\{9,10\}$ and edge $\{(9,10)\}$.
\end{example}




Note that $P_{(j,G)}$ is equivalent to $P_{(i,G)}$ if it satisfies that $j\in P_{(i,G)}$. Then, those edges adjacent to $P_{(i,G)}$ are represented as $E^P_{(i, G)}=\{(u,v)|(u,v)\in E, u\in P_{(i,G)} \lor v\in P_{(i,G)}\}$ and the following theorem could be deduced.


\begin{theorem}
  \label{the: cascade}
  The removal of an arbitrary edge within $E^P_{(i, G)}$ will absolutely make all nodes within $P_{(i, G)}$ collapse.
\end{theorem}
\begin{proof}

  According to Definition \ref{def: coronapedigree}, each node within $P_{(i,G)}$ possesses its core strength of $1$ which means the disconnection of any supportive neighbor will make this node collapse. Besides, each edge within $E^P_{(i,G)}$ actually bridges some corona node within $P_{(i,G)}$ with one of its supportive neighbors. In this way, if one of edges in $E^P_{(i,G)}$ is removed, the corona node (or corona nodes) adjacent to it will surely collapse. Because of the cascade phenomenon, the other nodes contained in $P_{(i,G)}$ will collapse follow.
\end{proof}
\begin{example}
  As shown in Figure \ref{fig: deletion}, taking $P_{(1,G)}$ where nodes $\{1,3\}$ exist as example, We get $E^P_{(1,G)}=\{(1, 2), (1, 3), (1, 5), (2, 3), (3, 7)\}$. Node $3$ will be absolutely squeezed out of $3$-core after the removal of an arbitrary edge contained in $E^P_{(1,G)}$, like $(2, 3)$, and then node $1$ will also collapse from $G_3$ because of the cascade phenomenon.
\end{example}

\begin{figure*}[!t]
  \centering
  \includegraphics[width=\linewidth]{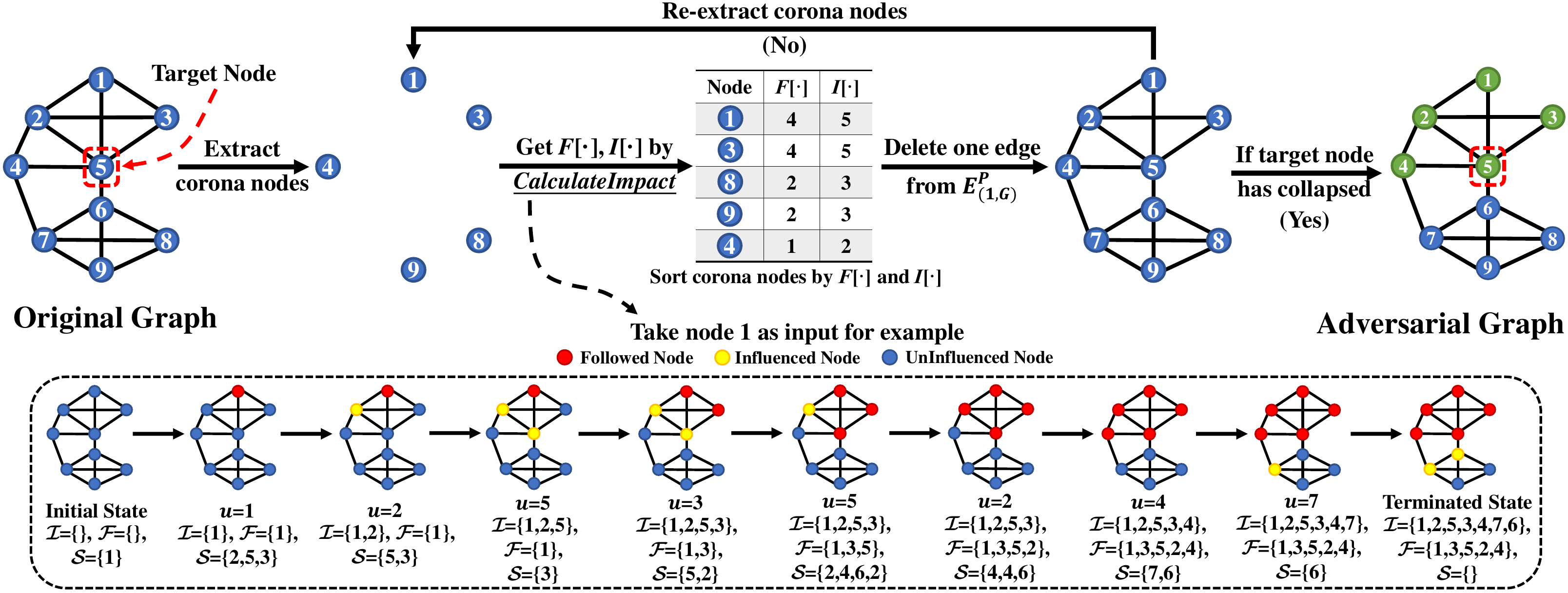}
  \caption{
  The framework of TNC algorithm. Given a target node $5$ with $CS_{(5,G)}=3$, we first extract the corona nodes from $G_3$, and then evaluate the impact of each corona node on the target node through CalculateImpact algorithm. After that, the most impacted node $1$ is filtered out and one edge existing in $E^P_{(1,G)}$ is deleted. If the target node has collapsed, the adversarial graph will be output and the removed edges will be returned; otherwise, we re-extract the corona nodes and repeat the above process. The contents shown in the dotted box display the detailed operations of CalculateImpact algorithm. The detailed descriptions will be presented in Section \ref{sec: TNC}.
  }
  \label{fig: progress}
\end{figure*}

\section{Methodologies}
\label{sec: methodology}


In this section, in order to address the TNCP problem, we propose an effective heuristic algorithm called Targeted $k$-Node Collapse (TNC) as the first solution. Additionally, based on TNC algorithm, we design an optimized strategy called Adjacent Targeted $k$-Node Collapse (ATNC) to further reduce computational complexity, making it suitable for large-scale networks.

\subsection{TNC Algorithm}
\label{sec: TNC}
To solve the TNCP problem, we propose the TNC algorithm, which itreatively removes one edge that can lead to the greatest impact on the target node $i$ until the target node collapses. The impact on the target node is determined by maximizing (i) the number of collapsed nodes within $SN_{(i,G)}$, and (ii) the number of nodes whose core strengths change within $SN_{(i,G)}$. As discussed earlier, those edges existing in $E_{k\setminus k+1}$ and connecting to $V^C_{(k,G)}$ play significant roles in the collapse of target node $i$ with $C_{(i,G)}=k$. Then, according to Theorem \ref{the: cascade}, for a corona node $u\in V^C_{(k,G)}$ and its corona pedigree $P_{(u,G)}\in G_k$, it is easy to realize that the disconnection of the relationship between node $u$ and one of its supportive neighbors will actually make all nodes within $P_{(u,G)}$ collapse from $G_k$ and then make all edges within $E^P_{(u,G)}$ be excluded from $E_{k\setminus k+1}$. In this manner, in order to avoid unnecessary duplicate operations, we only need to select the corona pedigree $P_{(v,G)}$ whose detachment leads to the greatest impact on the target node and removes one of edges existing in $E^P_{(v,G)}$ in each iteration until the target node collapses. Figure \ref{fig: progress} illustrates the overall framework of TNC algorithm along with the detailed operations shown in Algorithm \ref{alg: TNC}. Words for further descriptions are given as following.

As shown in Algorithm \ref{alg: TNC}, Line $4$, the corona nodes, $Coronas$, are firstly extracted from $G^\prime_k$ as candidates where $G^\prime$ is initialized as $G$ in Line $1$. After that, in Lines $6$-$7$, by exploiting an assistant algorithm called CalculateImpact which will be introduced in the following paragraphs, the impact which will be made on the target node $i$ is measured by $F[u]$ and $I[u]$ if corona node $u\in Coronas$ collapses, and $Coronas$ is updated according to Theorem \ref{the: cascade}. Next, we select the top corona node $v$ sorted according to $F[\cdot]$ (first priority) and $I[\cdot]$ (second priority) in Line $8$. Then, one of edges contained in $E^P_{(v,G^\prime)}$ is added to $e$ with the update of $G^\prime$ in Lines $13$-$14$. The above process will continue until there is the violation of Theorem \ref{the: coresupport} to make the target node $i$ collapse. Note that if the collapse of $v$ makes no supportive neighbors of node $i$ collapse, we will remove the edge bridging the target node and its supportive neighbor with the minimal core strength in $G^\prime$ as instead, in Lines $9$-$11$.

\textbf{CalculateImpact Algorithm.} After the collapse of node $u\in V$, for all nodes in $G$, those nodes whose core strength decreases are named as Influenced Nodes, those whose core value decreases are named as Followed Nodes, and those whose core strengths and core values remain the same are named as Uninfluenced Nodes. To effectively measure the impact that the collapse of a corona node $n$ can make on the target node $i$, we offer CalculateImpact algorithm which is based on Depth-First Search (DFS) and whose details are shown in Algorithm \ref{alg: DFS}. 

As shown in Algorithm \ref{alg: DFS}, Line $1$, $\mathcal{S}$ is defined to store the nodes waiting to be visited, $\mathcal{F}$ and $\mathcal{I}$ are defined to store the followed nodes and the influenced nodes, respectively. Besides, in Line $2$, a dictionary $\mathcal{T}$ with default value of $0$ is defined to record the decrease in the number of supportive neighbors of each node in $G$ after the input node $n$ collapses. In this way, for a visited node $u$ popped from $\mathcal{S}$, if $\mathcal{T}[u]>0$, it will be marked as an influenced node and be added into $\mathcal{I}$ in Line $7$; furthermore, if $CS_{(u,G)}\leq \mathcal{T}[u]$, it will also be marked as a followed node and be added into $\mathcal{F}$ in Line $9$. Besides, if node $u$ has been marked as a followed node, on the basis of Theorem \ref{lem: edgeinfluence}, those nodes contained in $\widetilde{SN}_{(u,G)}$ and satisfying $CS_{(\cdot,G)}>\mathcal{T}[\cdot]$ will be pushed into $\mathcal{S}$ in Line $10$. Please note that those nodes marked as followed nodes will be excluded from $\mathcal{S}$ in Line $11$. The above process will be repeated iteratively until $\mathcal{S}$ is empty.

For example, contents shown in the dotted box of Figure \ref{fig: progress} exhibit the detailed process of Algorithm \ref{alg: DFS} where node $1$ with $CS_{(1,G)}=1$ is taken as the input node. First, in the initial-state graph, $\mathcal{F}$ and $\mathcal{I}$ are initialized as empty sets and $\mathcal{S} = \{1\}$. Next, in the second graph, node $1$ is popped from $\mathcal{S}$ with the update of $\mathcal{T}[1]=1$ and be added into $\mathcal{I}$. It is apparent that node $1$ is also added into $\mathcal{F}$ because of the satisfaction of $CS_{(1,G)}\leq\mathcal{T}[1]$, and its neighbors $\{2,5,3\}$ are pushed into $\mathcal{S}$. After that, in the third graph, node $2$ with $CS_{(2,G)}=2$ is popped, and we get $\mathcal{T}[2]=1$ with the addition of node $2$ into $\mathcal{I}$. Then, the next iteration will be triggered directly because of $CS_{(2,G)}>\mathcal{T}[2]$. Continuing in this flow, we finally achieve that $\mathcal{I}=\{1, 2, 5, 3, 4, 7, 6\}$ and $\mathcal{F}=\{1, 3, 5, 2, 4\}$, and further get that $|\mathcal{N}_{(5,G)}\cap\mathcal{F}|=4$ and $|\mathcal{N}_{(5,G)}\cap\mathcal{I}|=5$.


\textbf{Time Complexity.} As shown in Algorithm \ref{alg: TNC}, first, in order to extract the corona nodes $Coronas$ of $G_k$ from $G$, it takes the time in the order of $\mathcal{O}(|V|)$ in Line $3$. Then, from Line $5$ to Line $7$, one corona node within each corona pedigree in $G_k$ is assigned weights through $CalculateImpact$ algorithm which takes the time in the order of $\mathcal{O}(\overline{CS}_{k\setminus k+1}\cdot |Coronas|)$ where $\overline{CS}_{k\setminus k+1}=\frac{\sum_{v\in V_{k}\setminus V_{k+1}}CS_{(v,G)}}{|V_{k}\setminus V_{k+1}|}$. After that, considering the worst condition, $CS_{(i,G)}$ iterations are executed with the total time complexity in the order of $\mathcal{O}(CS_{(i,G)}\cdot (|V| + \overline{CS}_{k\setminus k+1}\cdot |Coronas|))$.


\begin{algorithm}[ht]   
	\caption{TNC}         
  \label{alg: TNC}
  \SetKwInOut{Input}{input}
  \SetKwInOut{Output}{output}
	\Input{the given graph $G$, the target node $i$;}               
	\Output{the removed edges $e$.}              

  $e\leftarrow$ empty set; $G^\prime \leftarrow G$; $k\leftarrow C_{(i,G)}$\;
  $F,\ I\leftarrow$ dictionaries with default value of $0$\;
  \While{$|SN_{(i,k,G^\prime)}|\geq k$}
  {
    $Coronas\leftarrow\{u|u \in G^\prime_k, \mathcal{N}_{(u,G^\prime_k)}=k\}$\;
    \ForEach{$u\in Coronas$}
    {
      $F[u], \ I[u]\leftarrow \text{CalculateImpact} (G^\prime, i, u)$\;
      $Coronas\leftarrow Coronas\setminus P_{(u,G^\prime)}$\;
    }
    $v\leftarrow$ The top corona node sorted according to $F[\cdot]$ (first priority) and $I[\cdot]$ (second priority)\;
    
    \eIf{$F[v]=0$}
    {
      $m\leftarrow$ The supportive neighbor of $i$ in $G^\prime$ with the lowest core strength\;
      $e\leftarrow e\cup \{(i, m)\}$\;
    }
    {
      $e\leftarrow e\cup\{\forall(m, n)\in E^P_{(v, G^\prime)}\}$\;
    }
    $G^\prime\leftarrow G\setminus e$\;
  }
  \Return{$e$}
\end{algorithm}

\begin{algorithm}[ht]   
  \SetKwInOut{Input}{input}\SetKwInOut{Output}{output}
	\caption{CalculateImpact}         
  \label{alg: DFS}
	\Input{the given graph $G$, the target node $i$, the input node $n$;}               
  \Output{the number of followed nodes in $\mathcal{N}_{(i,G)}$, the number of influenced nodes in $\mathcal{N}_{(i,G)}$.}
  $\mathcal{S}\leftarrow$ empty stack; $\mathcal{F}\leftarrow$ empty set; $\mathcal{I}\leftarrow$ empty set\;
  $\mathcal{T}\leftarrow$ a dictionary with default value of $0$\;
  $\mathcal{S}.push(n)$\;
  \While{$\mathcal{S}$ is not empty}
  {
    $u\leftarrow\mathcal{S}.pop()$\;
    $\mathcal{T}[u]\leftarrow\mathcal{T}[u]+1$\;
    $\mathcal{I}\leftarrow\mathcal{I}\cup \{u\}$\;
    \If{$CS_{(u, G)}\leq\mathcal{T}[u]$}
    {
      $\mathcal{F}\leftarrow\mathcal{F}\cup\{u\}$\;
      $\mathcal{V}\leftarrow\{v|v\in \widetilde{SN}_{(u, G)}, CS_{(v, G)}>\mathcal{T}[v]\}$\;
      $\mathcal{S}.push(\mathcal{V})$\;
      $\mathcal{S}\leftarrow\mathcal{S}\setminus\mathcal{F}$\;
    }
  }
  \Return{$|\mathcal{N}_{(i,G)}\cap\mathcal{F}|$, $|\mathcal{N}_{(i,G)}\cap\mathcal{I}|$}
\end{algorithm}

\subsection{ATNC Algorithm}
\label{sec: ATNC}

In the previous part, we give the introduction of TNC algorithm which iteratively removes one edge that connected to the corona pedigree whose detachment could cause the greatest impact on the target node for addressing the TNCP problem. However, in each iteration, TNC algorithm needs to traverse all nodes within $G^\prime$ to extract the corona nodes of $G^\prime_k$ and then visit each corona pedigree through CalculateImpact algorithm to filter out the most impacted one. Clearly, the process is highly time-consuming for large-scale networks which pushes the expectation of a heuristic algorithm with less time complexity. In this part, we offer Adjacent Targeted $k$-Node Collapse (ATNC) improved from TNC which actually takes the strategy of adjacent search to exploit the local information of the target node. The details of ATNC are shown in Algorithm \ref{alg: ATNC} along with its descriptions as following. 


As shown in Algorithm \ref{alg: ATNC}, Line $3$, instead of extracting all corona nodes within $G^\prime_k$ by TNC algorithm, ATNC only exploits those corona nodes adjacent to the target node which are named as corona neighbors $CorNbrs$. Next, in Lines $5$-$8$, through the same operations as those of TNC, the top corona node $v$ is filtered out. After that, we add edge $(i,v)$ into $e$ with the update of $G^\prime$ and the re-extraction of $CorNbrs$ in Lines $9$-$10$. The above process will continue until $CorNbrs$ is empty or there is the violation of Theorem \ref{the: coresupport} for the target node $i$. Note that if the above loop quits with $SN_{(i,k,G^\prime)}>k$ which means that $|CorNbrs|=0$ and the target node still remains in $G_k$, then we will randomly sample $CS_{(i,G^\prime)}$ supportive neighbors from $SN_{(i,k,G^\prime)}$ and make the target node $i$ disconnected with them in Lines $12$-$14$.

\textbf{Time Complexity.} Similar to the time complexity of TNC, since only the corona nodes existing in the one-hop neighbors of the target node will be selected as candidates, the time for collecting the candidates is in the order of $\mathcal{O}(|\mathcal{N}_{(i,G)}|)$ in Algorithm \ref{alg: ATNC}, Line $3$ at first. Then, from Line $5$ to Line $7$, each corona pedigree contained in $G_k$ is traversed with the quantification of their impact to the target node which takes the time in the order of $\mathcal{O}(\overline{CS}_{(k\setminus k+1)} \cdot |CorNbrs|)$. After that, considering the worst condition, $CS_{(i,G)}$ iterations are executed with the total time complexity in the order of $\mathcal{O}(CS_{(i,G)}\cdot (|V| + \overline{CS}_{k\setminus k+1}\cdot |CorNbrs|))$.

\begin{algorithm}[ht]   
	\caption{ATNC}         
  \label{alg: ATNC}
	\KwIn{the given graph $G$, the target node $i$;}               
	\KwOut{the removed edges $e$.}              
  $e\leftarrow$ empty set; $G^\prime\leftarrow  G$; $k\leftarrow C_{(i,G)}$\;
  $F,\ I\leftarrow$ dictionaries with default value of $0$\;
  $CorNbrs\leftarrow \{j|j\in \widetilde{SN}_{(i,G^\prime)}, CS_{(j,G^\prime)}=1\}$\;
  \While{$|CorNbrs|>0$ \textbf{and} $|SN_{(i,k,G^\prime)}|\geq k$}
  {
    \ForEach{$u\in CorNbrs$}
    {
      $F[u], \ I[u]\leftarrow \text{CalculateImpact} (G^\prime, i, u)$\;
      $CorNbrs\leftarrow CorNbrs\setminus P_{(u,G^\prime)}$\;
    }
    
    $v\leftarrow$ The top corona node sorted according to $F[\cdot]$ (first priority) and $I[\cdot]$ (second priority)\;
    $e\leftarrow e\cup \{(i,v)\}$\;
    $G^\prime\leftarrow G\setminus e$\;
    Re-extract $CorNbrs$\;
  }
  \If{$|SN_{(i,k,G^\prime)}|\geq k$}
  {
    $e^\prime\leftarrow \text{Sample}(\{(i,j)|j\in SN_{(i,G^\prime)}\}, CS_{(i,G^\prime)})$\;
    $e\leftarrow e\cup e^\prime$\;
  }
  \Return{$e$}
\end{algorithm}

\section{Experiments}
\label{sec: experiment}

In this section, our experiments will be conducted on $16$ real-world network datasets collected from various domains to demonstrate the performance of TNC and ATNC. We also include $4$ baseline methods for comparisons. All of our experiments are deployed on a server with Intel(R) Xeon(R) Gold 5218R CPU @ 2.10GHz and 377GB RAM, which installs Linux Ubuntu 20.04.4.


\subsection{Datasets}
The basic properties of $16$ real-world networks from various domains, e.g., Social Network (SN), Collaboration Network (CN), Infrastructure Network (IN) and Web Network (WN), are presented in Table \ref{tab: dataset}. Different labels are exploited to distinguish the different public platforms where networks are collected. For example, those marked with stars are collected from \url{https://networkrepository.com/} \cite{nr} and those marked with circles are collected from \url{http://snap.stanford.edu/} \cite{snapnets}. Please note that all networks used in the following experiments are converted to undirected and unweighted graphs, with no self-loops or isolated nodes. Due to the space limitation, more detailed information of these networks could be achieved on the mentioned websites.

\begin{table}[thbp]
  \caption{Basic properties of mentioned networks containing the number of nodes $|V|$, the number of edges $|E|$, the maximal value of $k$-core $k_{max}$ and the average degree $d_{avg}$.}
  \label{tab: dataset}
  \begin{tabular*}{\hsize}{@{}@{\extracolsep{\fill}}ll|cccc@{}}
    \bottomrule[0.5mm]
    \multicolumn{2}{c}{Network}                                & $|V|$   & $|E|$   & $k_{max}$  & $d_{avg}$ \\\hline
    \multirow{3}{*}{SN}                     & TVShow$^\star$   & 3892    & 17239   & 56         & 8.8587    \\
                                            & LastFM$^\circ$   & 7624    & 27806   & 20         & 7.2943    \\
                                            & Facebook$^\circ$ & 22470   & 170823  & 56         & 15.2045   \\
                                            & DeezerEU$^\circ$ & 28281   & 92752   & 12         & 6.5593    \\
                                            & Gowalla$^\circ$  & 196591  & 950327  & 51         & 9.6681    \\ \hline
    \multirow{3}{*}{CN}                     & HepPh$^\star$    & 12006   & 118489  & 238        & 19.7383   \\
                                            & AstroPh$^\star$  & 18771   & 198050  & 56         & 21.1017   \\
                                            & CondMat$^\star$  & 21363   & 91286   & 25         & 8.5462    \\
                                            & Citeseer$^\star$ & 227320  & 814134  & 86         & 7.1629    \\ \hline
    \multirow{3}{*}{IN}                     & USAir$^\star$    & 332     & 2126    & 26         & 12.8072   \\
                                            & USPower$^\star$  & 4941    & 6594    & 5          & 2.6691    \\
                                            & RoadNet$^\star$  & 1965206 & 2766607 & 3          & 2.8156    \\ \hline
    \multirow{3}{*}{WN}                     & EDU$^\star$      & 3031    & 6474    & 29         & 4.2719    \\
                                            & Indo$^\star$     & 11358   & 47606   & 49         & 8.3828    \\
                                            & Arabic$^\star$   & 163598  & 1747269 & 101        & 21.3605   \\
                                            & Google$^\circ$   & 875713  & 4322051 & 44         & 9.8709    \\
    \toprule[0.5mm]
  \end{tabular*}
\end{table}

\subsection{Baselines}
Given that we are the first work to study the TNCP problem, there is no ready-made method that can be used as a comparison experiment. For this reason, we design two random-based baseline methods and adjust two existing algorithms which are originally proposed to solve the $k$-core minimization problem. Their details are shown as follows.

\begin{itemize}[itemsep= 0 pt, topsep = 0 pt, itemindent = 0 pt, leftmargin = 15 pt]
  \item \textbf{Random Edge Deletion (RED)} arbitrarily selects an edge within $E$ to remove and then updates the core values of nodes within $V$. These two steps will be performed iteratively until the target node collapses successfully.
  \item \textbf{Random Neighbor Disconnection (RND)} arbitrarily removes an edge connected to the target node and then updates the core values of nodes within $V$. These two steps will be performed iteratively until the target node collapses successfully.
  \item \textbf{KNM} was proposed by \cite{chen2021edge} as a solution to the $k$-core minimization problem. It works by iteratively removing the edge whose detachment will lead to the maximal number of nodes who collapse from $G_k$. This process continues until the perturbation budget is reached or $G_k=\emptyset$. In this paper, we adapt the termination condition of KNM algorithm to the collapse of the target node.
  \item \textbf{SV} was proposed by \cite{ijcai2020p480} for covering the $k$-core minimization problem which exploits the shapley value, a cooperative game-theoretic concept. It assigns weights to the candidate edges and then chooses the top $b$ edges to remove. In this paper, considering the consumption of time, we set $E_{k\setminus k+1}$ as the candidate edges instead of $E_k$ which is originally used by \cite{ijcai2020p480}, and we set the hyperparameter $\epsilon^2=0.1$. Then we remove candidate edges one by one according to their weights until the target node collapses without the budget limitation of $b$.
\end{itemize}




In order to evaluate the transferability not only among various networks but also among various individual nodes, we will apply all baseline methods as well as our proposed algorithms on each node within every network to achieve its node robustness. Then we will evaluate the effectiveness of these algorithms by several global metrics which will be introduced in Section \ref{sec: metrics}. Additionally, it is necessary to be noted that both of RED and RND will be performed $10$ times independently on each node in order to reduce the randomness and the mean value is recorded as the robustness of each node.

\subsection{Metrics}
\label{sec: metrics}

We propose the following metrics, \emph{Number of Bubble Nodes} (NBN), \emph{Sum of Reduced Cost} (SRC), \emph{Weighted Average Reduction} (WAR), and \emph{Reduction Proportion} (RP) to evaluate the effectiveness of various methods.

\begin{itemize}[itemsep= 0 pt, topsep = 0 pt, itemindent = 0 pt, leftmargin = 15 pt]
  \item \textbf{NBN:} Through a particular algorithm, we are interested in how many bubble nodes can be explored from $G$. Thus, the total number of explored bubble nodes is recorded as NBN which is formulated as below:
  \begin{equation}
    \text{NBN} = |BN|.
  \end{equation}
 The higher NBN is, the more transferable the algorithm is among various nodes in a graph.

  \item \textbf{SRC:} For a bubble node $i$, the decrease between its core strength and node robustness is named as Reduced Cost which is quantified as $RC_{(i,G)}=CS_{(i,G)}-NR_{(i,G)}$. Therefore, the sum of reduced cost of all explored bubble nodes in $G$ could be formulated as below:
  \begin{equation}
    \text{SRC} = \sum_{i\in BN}{RC_{(i,G)}}.
  \end{equation}

  \item \textbf{WAR:} In order to illustrate the average cost reduction of explored bubble nodes in a network through some algorithm, we propose WAR which is formulated as below:
  \begin{equation}
    \text{WAR} = \frac{\sum\limits_{r\in \mathcal{U}}p_r^{-1}\cdot r}{\sum\limits_{r\in \mathcal{U}}p_r^{-1}}.
  \end{equation}
  where $\mathcal{U}$ contains the unique elements of $\{RC_{(i,G)}|i\in BN\}$ and $p_r = \frac{|\{i|i\in BN, RC_{(i,G)}=r\}|}{|BN|}$ denotes the probability of those nodes whose reduced cost equal to $r$ appearing in $BN$. And the reason why we do not use arithmetic average will be explained in Section \ref{sec: evaluation} with specific examples.
  
  \item \textbf{RP:} We are also interested in the reduction proportion of node robustness relative to core strength on all nodes in $G$ and propose RP for measuring, which is formulated as below:
  \begin{equation}
    \text{RP} = \frac{\text{SRC}}{\sum_{i\in V}CS_{(i,G)}}\times 100\%.
  \end{equation}
  In addition, RP can be used to describe the redundancy of core strength with respect to node robustness. The higher the RP, the more redundant the core strength.
\end{itemize}

\begin{table}[thbp]
  \caption{The experimental results of RED and RND. Those datasets that could not be covered within $10^5$ seconds are marked as /. Attention that for RED and RND, the robustness of each node is assigned as the average of the results achieved by $10$ independent experiments.}
  \label{tab: table1}
  \setlength{\tabcolsep}{0.5mm}
  \begin{tabular*}{\hsize}{@{}@{\extracolsep{\fill}}c|cccc|cccc@{}}
    \bottomrule[0.5mm]
    \multirow{2}{*}{Network} & \multicolumn{4}{c}{RED}           & \multicolumn{4}{c}{RND}            \\ \cline{2-9} 
    & NBN     & SRC    & WAR    & RP(\%) & NBN   & SRC      & WAR     & RP(\%) \\ \hline
    TVShow                   & 8   & 1.8   & 0.2    & 0.02   & 375   & 204.9   & 1.6894  & 2.75  \\
    LastFM                   & 23  & 20.1  & 1.2346 & 0.15   & 263   & 350.4   & 5.2813  & 2.7   \\
    Facebook                 & 9   & 4.7   & 0.4083 & 0.01   & 1346  & 1022.5  & 4.3603  & 2.0   \\
    DeezerEU                 & 1   & 0.4   & 0.4    & 0.001  & 648   & 324.8   & 2.3962  & 0.62  \\
    Gowalla                  & /   & /     & /      & /      & 5215  & 3015.2  & 4.1356  & 1.02  \\ \hline
    HepPh                    & 159 & 308.9 & 5.4106 & 1.51   & 1181  & 2773.4  & 13.3372 & 13.59 \\
    AstroPh                  & 94  & 90.2  & 2.1622 & 0.24   & 2234  & 4337.7  & 10.7611 & 11.48 \\
    CondMat                  & 13  & 2.4   & 0.2441 & 0.007  & 2439  & 2209.4  & 4.3676  & 6.05  \\
    Citeseer                 & /   & /     & /      & /      & 23285 & 20848.6 & 12.4333 & 6.37  \\ \hline
    USAir                    & 2   & 0.3   & 0.15   & 0.05   & 10    & 3.9     & 0.4714  & 0.66  \\
    USPower                  & 1   & 0.4   & 0.4    & 0.005  & 97    & 50.2    & 0.9693  & 0.62  \\
    RoadNet                  & /   & /     & /      & /      & /     & /       & /       & /     \\ \hline
    EDU                      & 3   & 0.4   & 0.1667 & 0.005  & 14    & 11.9    & 0.9455  & 0.16  \\
    Indo                     & 20  & 3.7   & 0.7049 & 0.02   & 754   & 703.4   & 3.8977  & 4.68  \\
    Arabic                   & 9   & 1.8   & 0.3032 & 0.0009 & 4510  & 5086.6  & 6.9982  & 2.42  \\
    Google                   & /   & /     & /      & /      & 53222 & 50569.3 & 18.2091 & 3.1   \\
    \toprule[0.5mm]
  \end{tabular*}
\end{table}

\begin{table*}[thbp]
  \caption{The experiment results of KNM, SV, TNC and ATNC. Those datasets that could not be completed by the method within $10^5$ seconds are marked as /. The best results are bolded and the second-best results are underlined.}
  \label{tab: table2}
  \setlength{\tabcolsep}{1mm}
  \begin{tabular*}{\hsize}{@{}@{\extracolsep{\fill}}c|cccc|cccc|cccc|cccc@{}}
    \bottomrule[0.5mm]
    \multirow{2}{*}{Network} & \multicolumn{4}{c}{KNM}        & \multicolumn{4}{c}{SV}                                  & \multicolumn{4}{c}{TNC}                                            & \multicolumn{4}{c}{ATNC}                                             \\ \cline{2-17} 
    & NBN  & SRC   & WAR    & RP(\%) & NBN       & SRC         & WAR            & RP(\%)       & NBN           & SRC             & WAR             & RP(\%)         & NBN            & SRC              & WAR             & RP(\%)         \\ \hline
    TVShow                   & 188  & 502  & 8.8894  & 6.73  & 142       & 403        & 8.0841          & 5.40         & \textbf{543}  & \textbf{1104}  & \textbf{10.4244} & \textbf{14.80} & {\ul 514}      & {\ul 1055}      & {\ul 10.4196}    & {\ul 14.14}    \\
    LastFM                   & 201  & 1025 & 17.0836 & 7.90  & 164       & 895        & 16.4166         & 6.90         & \textbf{458}  & \textbf{1652}  & \textbf{17.7163} & \textbf{12.74} & {\ul 404}      & {\ul 1485}      & {\ul 17.6044}    & {\ul 11.45}    \\
    Facebook                 & 828  & 5335 & 35.8553 & 10.42 & 618       & 3040       & 28.4149         & 5.90         & \textbf{2709} & \textbf{11691} & \textbf{40.5852} & \textbf{22.83} & {\ul 2468}     & {\ul 10033}     & {\ul 37.4845}    & {\ul 19.59}    \\
    DeezerEU                 & 182  & 606  & 11.3760 & 1.15  & 94        & 303        & 9.139           & 0.58         & \textbf{1162} & \textbf{2185}  & \textbf{15.4195} & \textbf{4.15}  & {\ul 1062}     & {\ul 1967}      & {\ul 15.1067}    & {\ul 3.74}     \\
    Gowalla                  & /    & /    & /       & /     & {\ul 828} & {\ul 9068} & {\ul 58.8298}   & {\ul 3.06}   & /             & /              & /                & /              & \textbf{7719}  & \textbf{22864}  & \textbf{62.2077} & \textbf{7.72}  \\ \hline
    HepPh                    & 562  & 3664 & 31.4897 & 17.96 & 481       & 3244       & 31.4435         & 15.90        & \textbf{1381} & \textbf{5142}  & \textbf{31.9080} & \textbf{25.19} & {\ul 1332}     & {\ul 5026}      & {\ul 31.7788}    & {\ul 24.63}    \\
    AstroPh                  & 1276 & 7702 & 30.2089 & 20.38 & 816       & 5191       & 29.0587         & 13.73        & \textbf{2984} & \textbf{12637} & \textbf{32.9459} & \textbf{33.43} & {\ul 2727}     & {\ul 11382}     & {\ul 31.6881}    & {\ul 30.11}    \\
    CondMat                  & 259  & 797  & 12.2096 & 2.18  & 184       & 574        & 10.3739         & 1.57         & \textbf{3008} & \textbf{6760}  & \textbf{13.1719} & \textbf{18.50} & {\ul 2801}     & {\ul 6231}      & {\ul 13.0097}    & {\ul 17.05}    \\
    Citeseer                 & /    & /    & /       & /     & {\ul 514} & {\ul 1884} & {\ul 18.377}    & {\ul 0.58}   & /             & /              & /                & /              & \textbf{25244} & \textbf{50203}  & \textbf{29.3033} & \textbf{15.34} \\ \hline
    USAir                    & 30   & 111  & 3.8897  & 18.91 & 21        & 59         & 2.6262          & 10.05        & \textbf{36}   & \textbf{120}   & \textbf{3.9747}  & \textbf{20.44} & \textbf{36}    & \textbf{120}    & \textbf{3.9747}  & \textbf{20.44} \\
    USPower                  & 21   & 27   & 2.7141  & 0.33  & 16        & 20         & 2.5852          & 0.25         & \textbf{109}  & \textbf{130}   & \textbf{2.9294}  & \textbf{1.60}  & \textbf{109}   & {\ul 127}       & {\ul 2.9212}     & {\ul 1.56}     \\
    RoadNet                  & /    & /    & /       & /     & {\ul 19}  & {\ul 21}   & \textbf{2.8945} & {\ul 0.0006} & /             & /              & /                & /              & \textbf{4105}  & \textbf{4175}   & {\ul 2.8658}     & \textbf{0.11}  \\ \hline
    EDU                      & 11   & 16   & 2.5805  & 0.21  & 12        & 16         & 2.5516          & 0.21         & \textbf{59}   & \textbf{67}    & \textbf{2.8268}  & \textbf{0.89}  & \textbf{59}    & \textbf{67}     & \textbf{2.8268}  & \textbf{0.89}  \\
    Indo                     & 73   & 157  & 9.4433  & 1.04  & 69        & 140        & 9.1799          & 0.93         & \textbf{779}  & \textbf{1212}  & {\ul 12.4205}    & \textbf{8.06}  & {\ul 772}      & {\ul 1177}      & \textbf{12.6278} & {\ul 7.83}     \\
    Arabic                   & 354  & 550  & 7.3287  & 0.26  & 214       & 486        & 8.8685          & 0.23         & \textbf{5088} & \textbf{8232}  & {\ul 14.4142}    & \textbf{3.92}  & {\ul 5001}     & {\ul 8049}      & \textbf{14.423}  & {\ul 3.83}     \\
    Google                   & /    & /    & /       & /     & {\ul 868} & {\ul 4309} & {\ul 35.6256}   & {\ul 0.26}   & /             & /              & /                & /              & \textbf{77928} & \textbf{222404} & \textbf{75.7784} & \textbf{13.63}\\
    \toprule[0.5mm]
  \end{tabular*}
\end{table*}

\begin{figure*}[t]
  \centering
	\includegraphics[width=\linewidth]{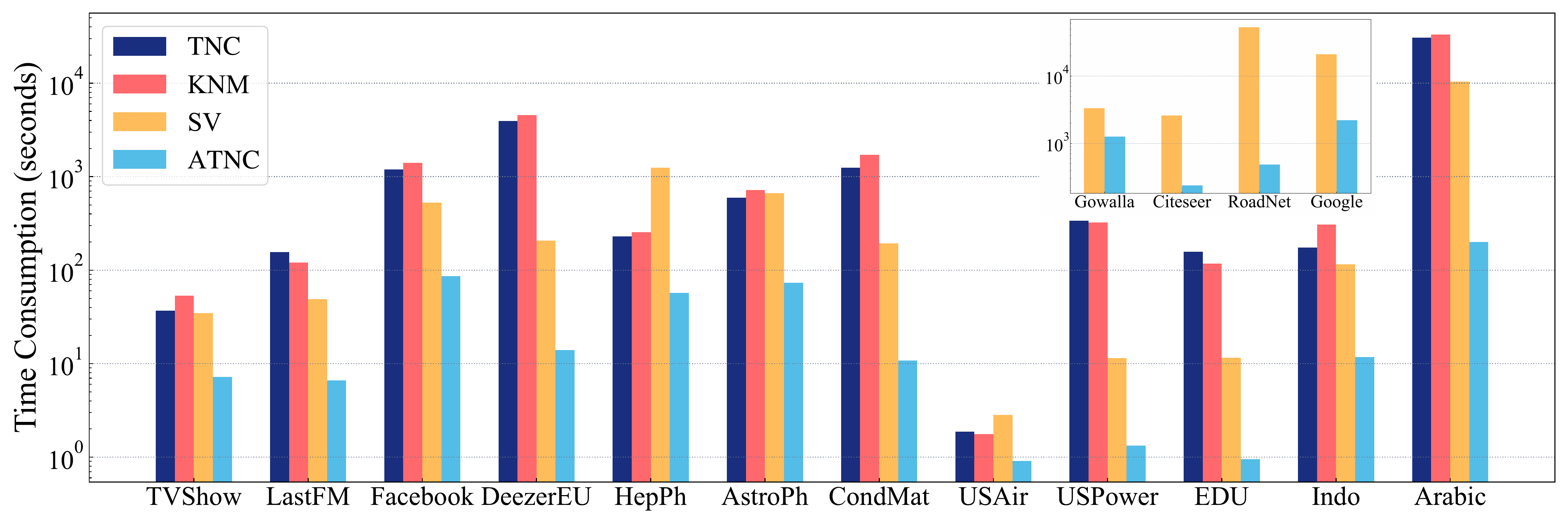}
  \vspace{-5mm}
  \caption{Comparisons of the running time of TNC, KNM, SV and ATNC. The subfigure exhibits the running time of SV and ATNC on those large-scale networks separately. We can see that ATNC is significantly more efficient than the other methods on the time consumption.}
  \label{fig: time_all}
\end{figure*}

\section{Results and Analyses}
\label{sec: results}


The experimental results are exhibited in Table \ref{tab: table1} and Table \ref{tab: table2}, which contrastively shows the performance of TNC and ATNC, compared with $4$ baseline methods on $16$ real-world networks mentioned before. Meanwhile, the detailed comparisons and analyses are presented as follows.

\subsection{Performance Evaluation}
\label{sec: evaluation}


\textbf{Comparisons Among Baselines.} Let us concentrate on Table \ref{tab: table1} in which the experimental results of RED and RND are exhibited. Notice that the robustness of each node in networks is achieved by the average of $10$ independent experimental results. From the table, it is easy to find that the NBN of RED is far fewer than that of RND on all networks which represents that RED is unable to explore bubble nodes and fails to cover the TNCP problem. On the contrary, RND performs much better than RED on all used metrics which demonstrates that the strategy of adjacent search for candidate reduction is helpful for covering the TNCP problem. Additionally, we can realize that RED is not able to complete the search missions on the $4$ networks whose number of nodes is more than $10^5$, i.e. Gowalla, Citeseer, RoadNet and Google, while RND only fails on RoadNet, a network with millions of nodes. It also proves that the strategy of adjacent search can effectively reduce the time complexity of the algorithm.

After that, let us turn our sights to the experimental results achieved by KNM and SV which are illustrated in Table \ref{tab: table2}. Neither KNM nor SV displays powerful transferability among different nodes in a network compared to RND. For instance, RND detects $2439$ bubble nodes on CondMat network, whereas this number is $259$ and $184$ induced by KNM and SV, respectively, which reveals a difference of almost $10$ times. Similarly, RND is able to filter out $4510$ bubble nodes on Arabic network, while KNM and SV could only find $354$ and $214$ nodes. However, the other metrics, i.e., SRC, WAR and RP, are much higher for KNM and SV compared to those for RND. For example, on Facebook network, the SRC of KNM is $5$ times larger than that of RND and on AstroPh network, the WAR of KNM is $3$ times larger than that of RND. These results tell us that the heuristic methods enable the target node to collapse at a lower budget compared to the random-based methods, although they can only work on part of bubble nodes. Analysis from the principle of these two algorithms, neither of them exploits the information associated with the target node to guide the removal of edges which leads to the unsatisfied performance on solving the TNCP problem.

\begin{figure*}[t]
	\centering
	\subfigure[LastFM, Node 6101]{
	\includegraphics[scale=0.29]{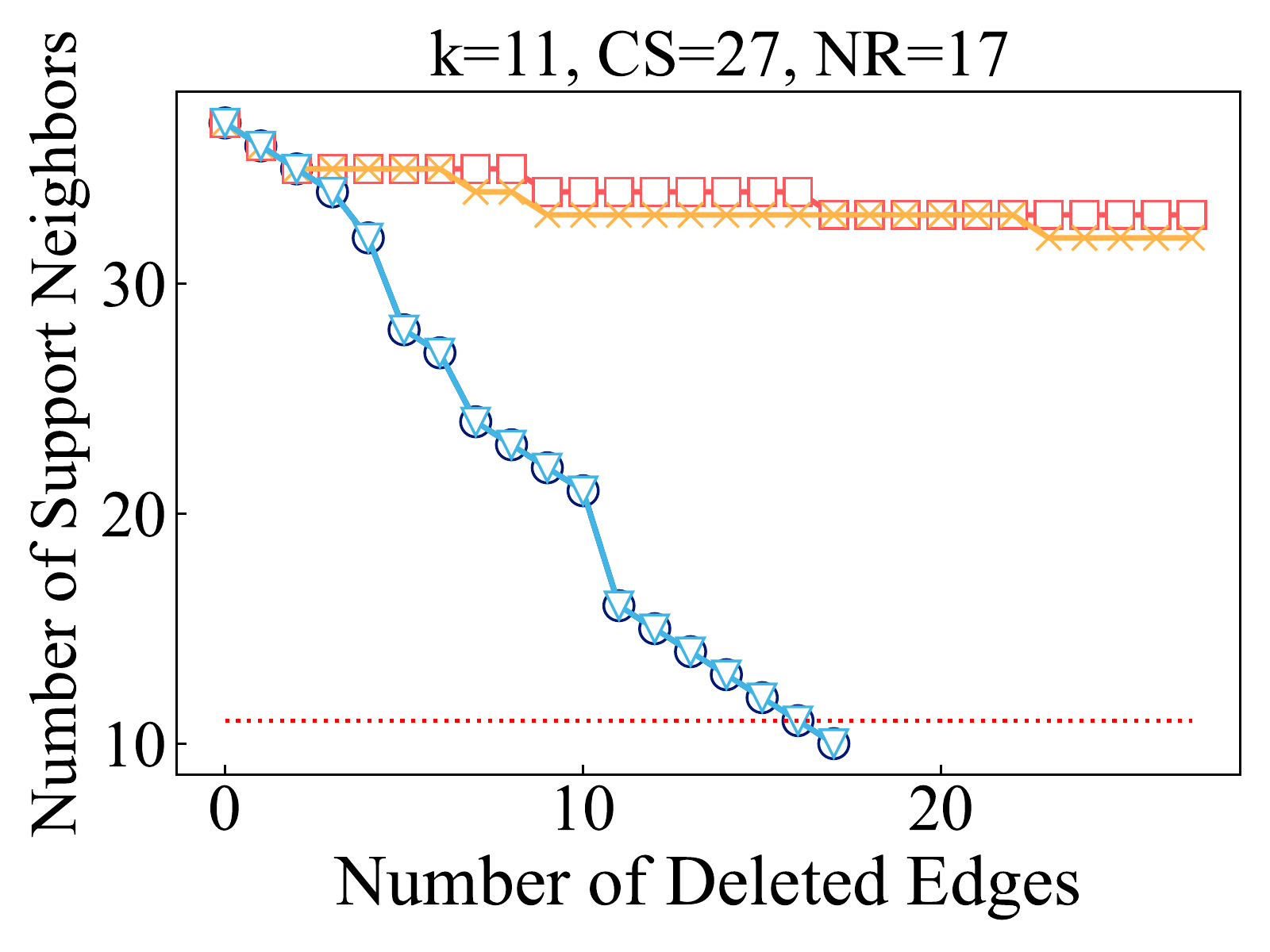}\label{fig: case_a}}
	\hspace{-2.5mm}
	\subfigure[LastFM, Node 3103]{
	\includegraphics[scale=0.29]{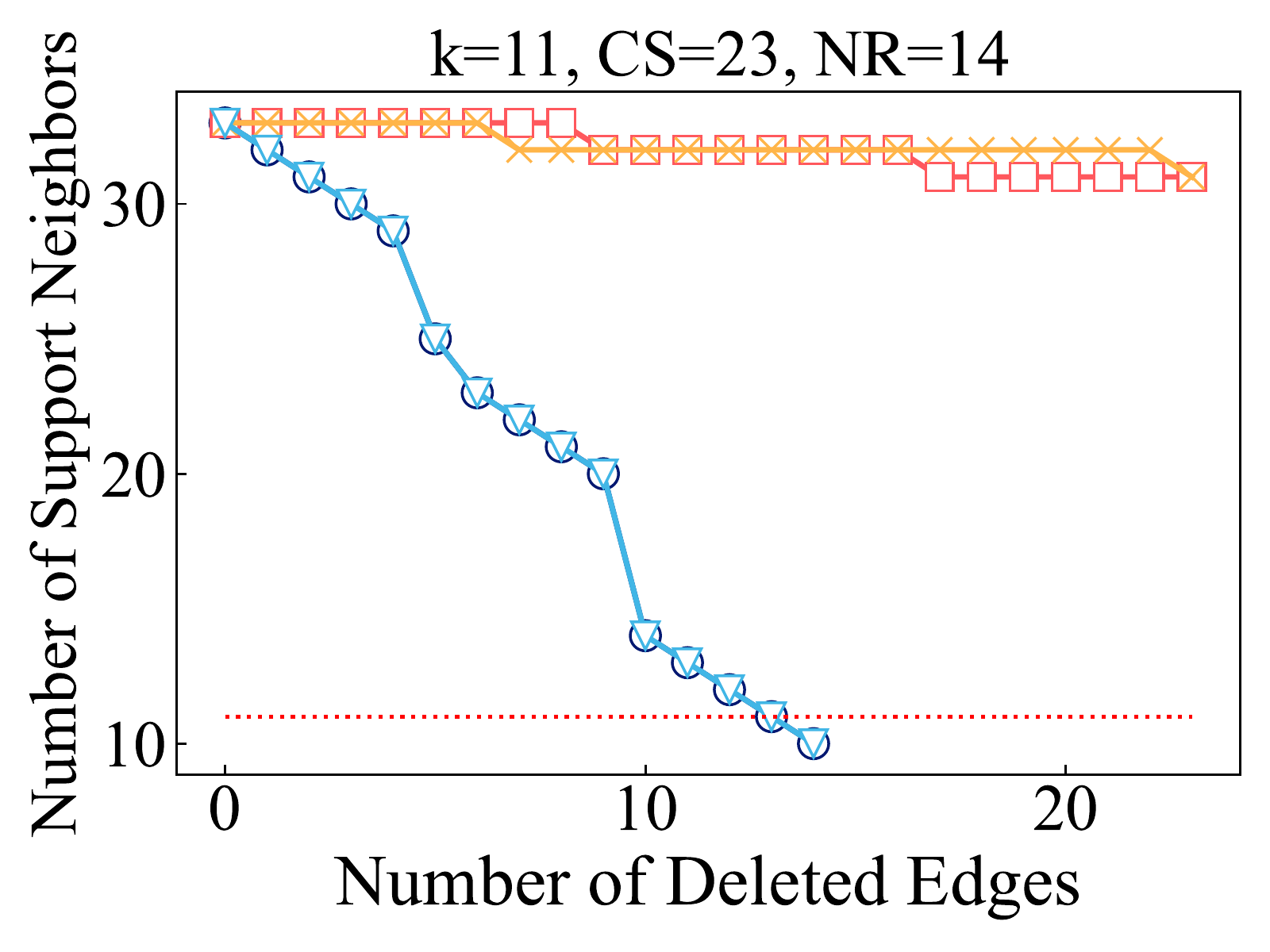}\label{fig: case_b}}
	\hspace{-2.5mm}
  \subfigure[DeezerEU, Node 17963]{
	\includegraphics[scale=0.29]{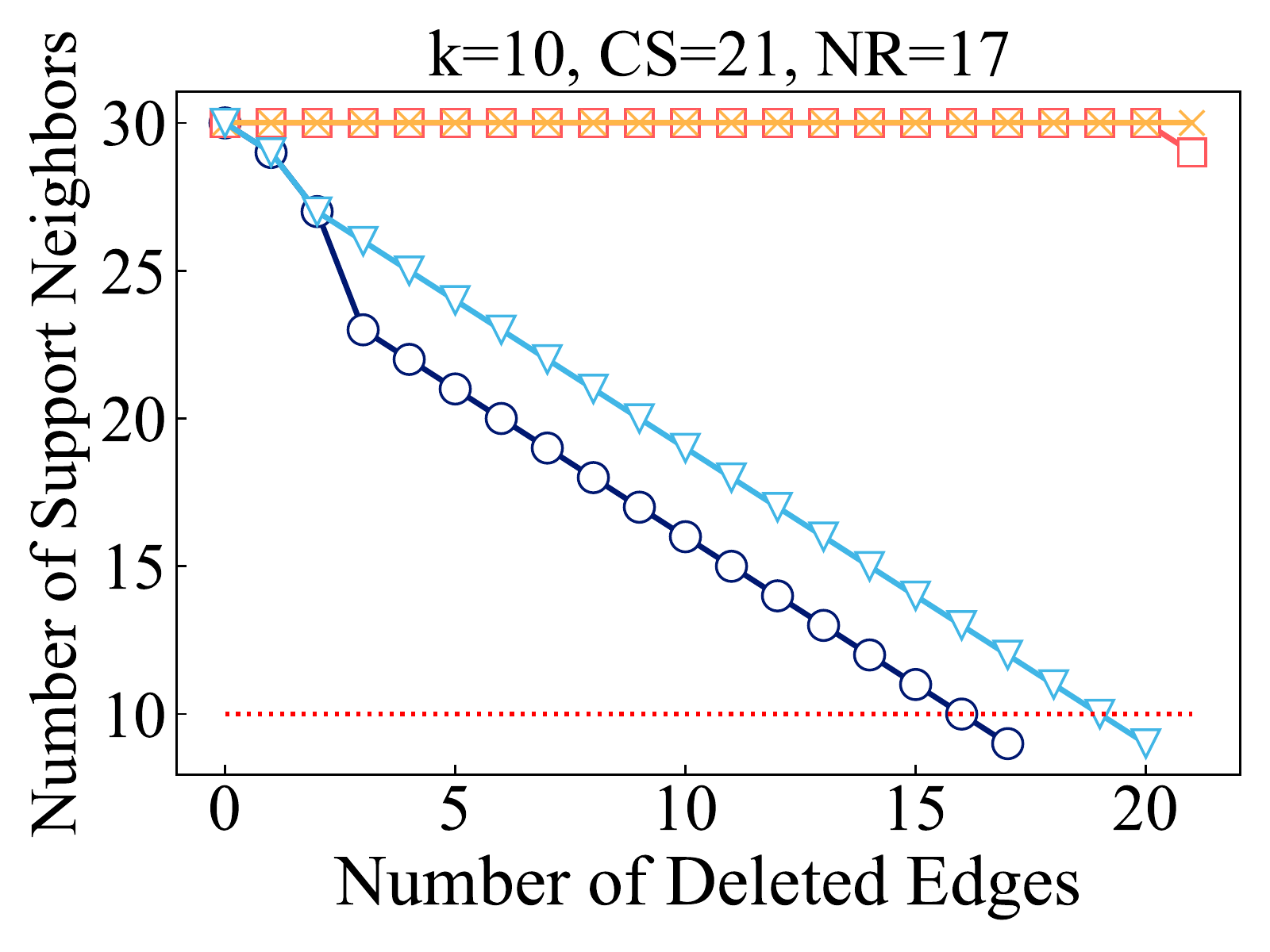}\label{fig: case_c}}
	\hspace{-2.5mm}
  \subfigure[DeezerEU, Node 24062]{
	\includegraphics[scale=0.29]{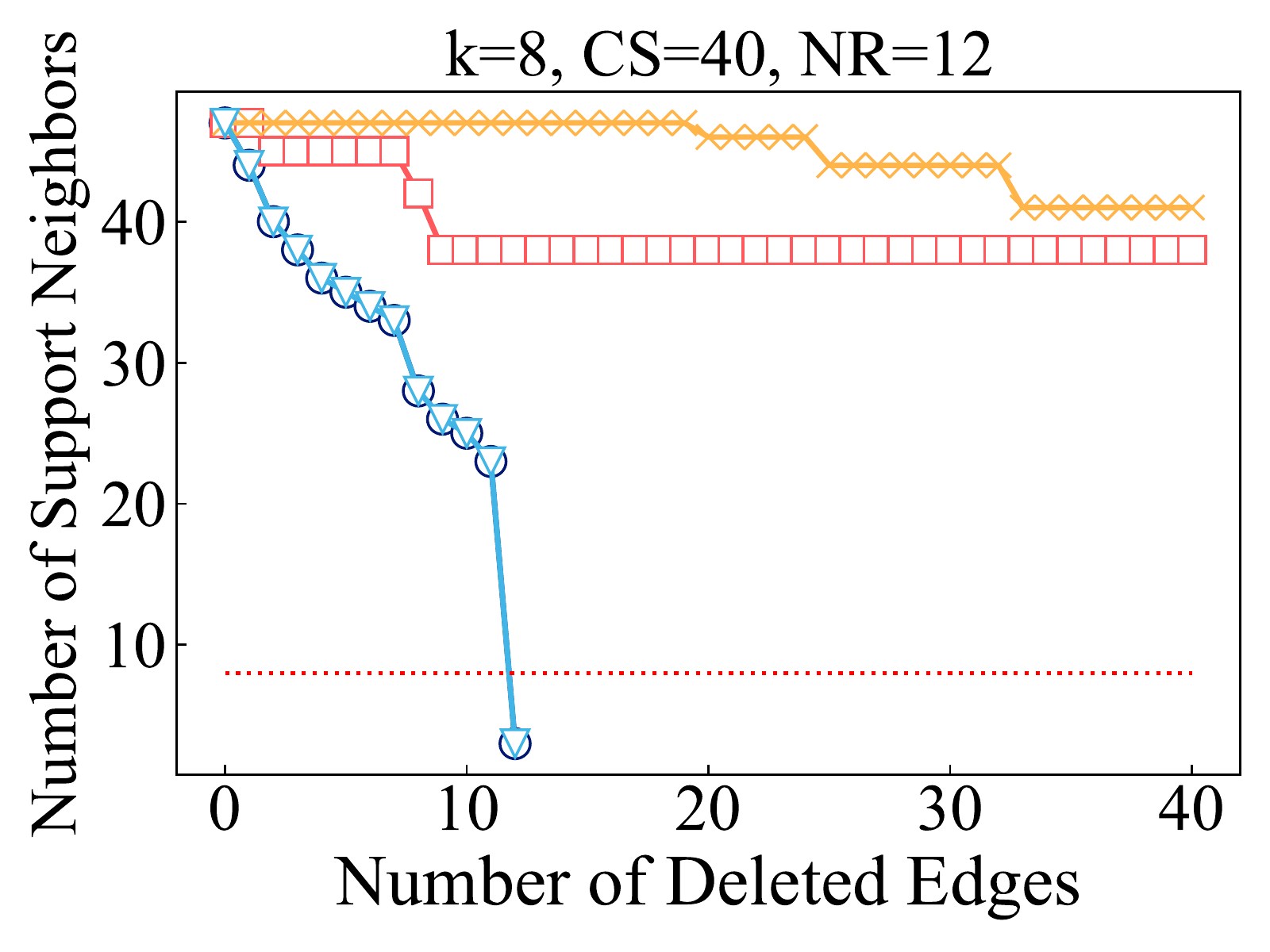}\label{fig: case_d}}\\
  
	\subfigure[CondMat, Node 1233]{
	\includegraphics[scale=0.29]{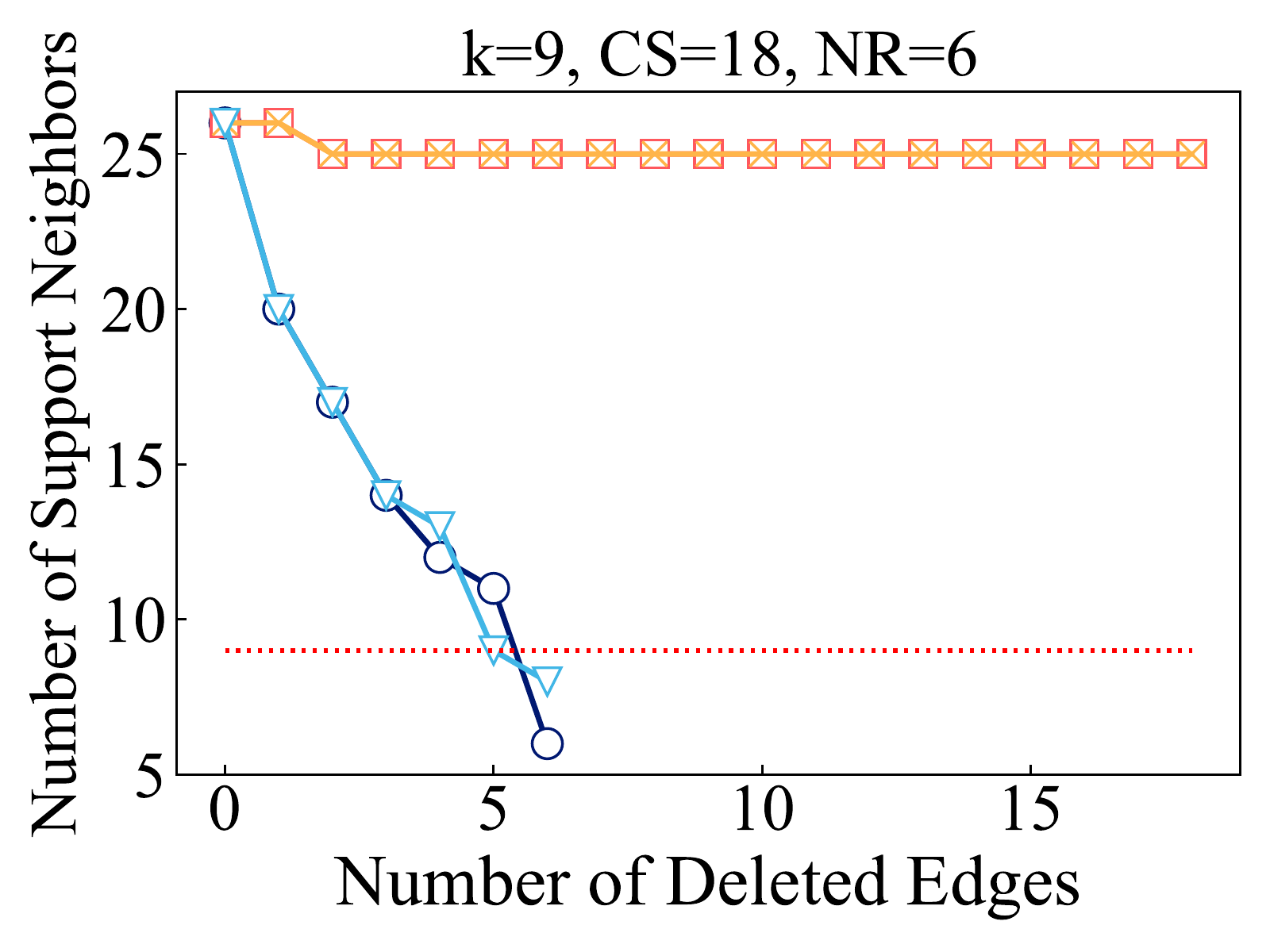}\label{fig: case_e}}
	\hspace{-2.5mm}
	\subfigure[CondMat, Node 13621]{
	\includegraphics[scale=0.29]{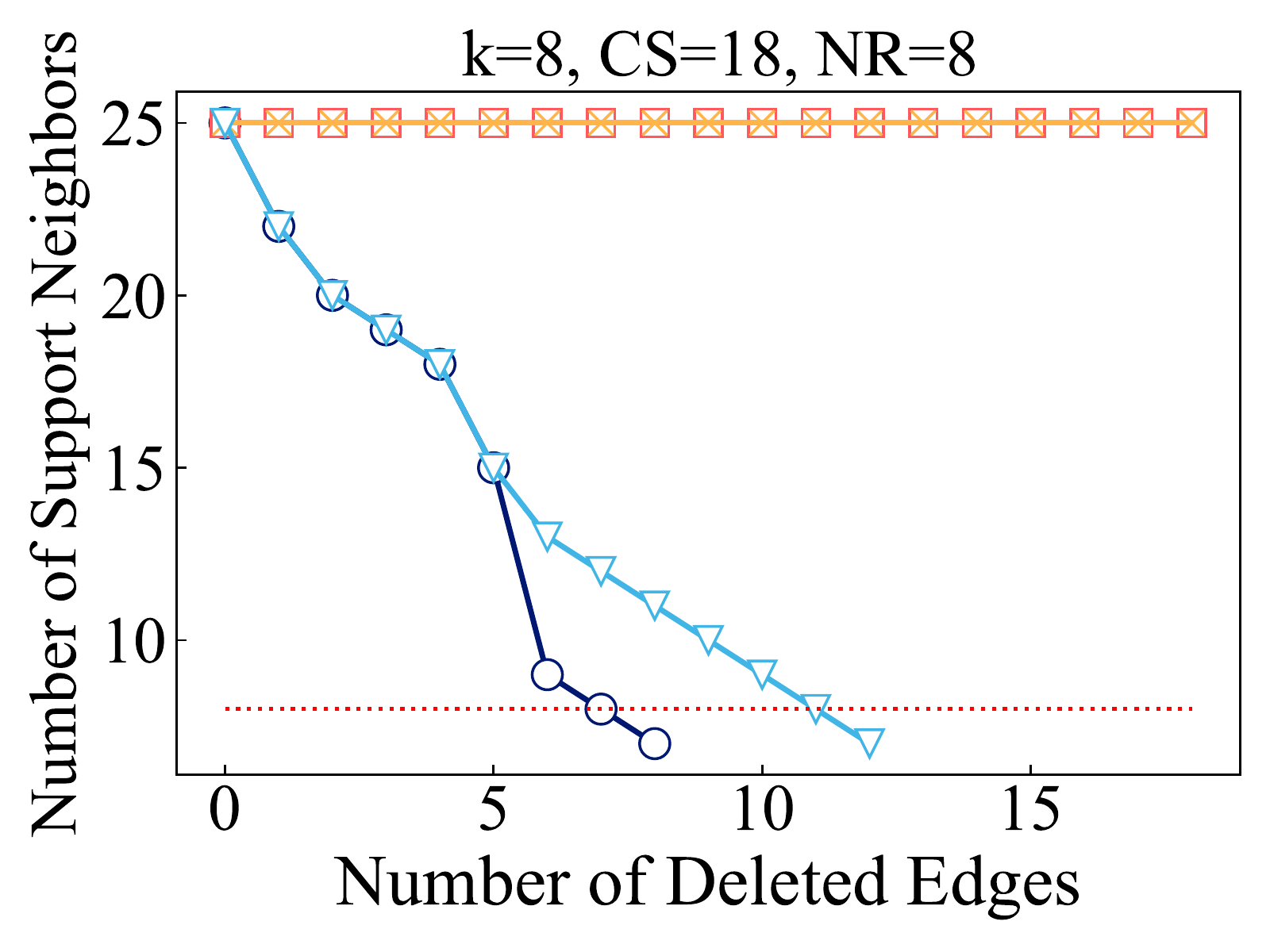}\label{fig: case_f}}
	\hspace{-2.5mm}
  \subfigure[Indo, Node 2721]{
	\includegraphics[scale=0.29]{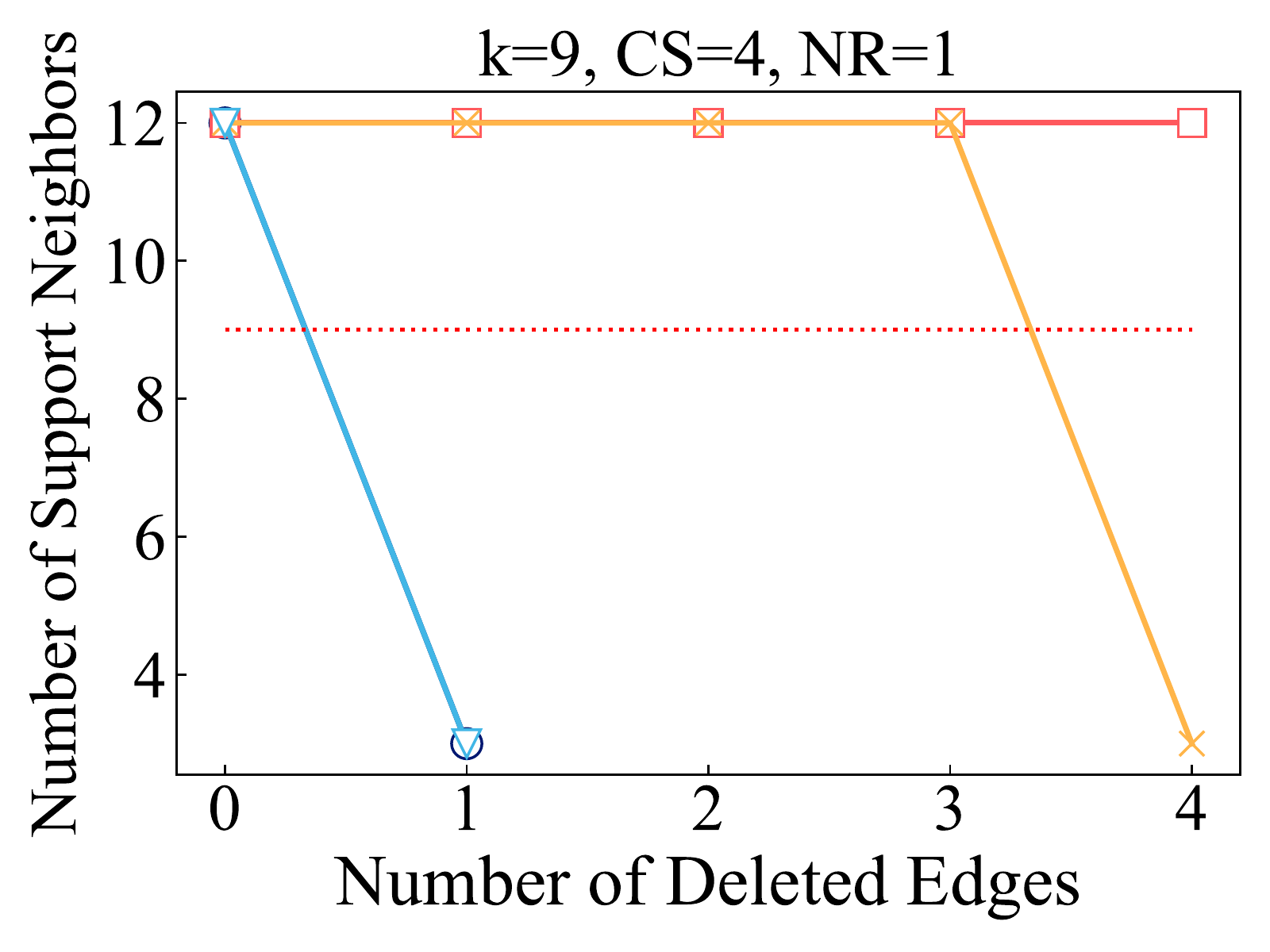}\label{fig: case_g}}
	\hspace{-2.5mm}
  \subfigure[Indo, Node 4712]{
	\includegraphics[scale=0.29]{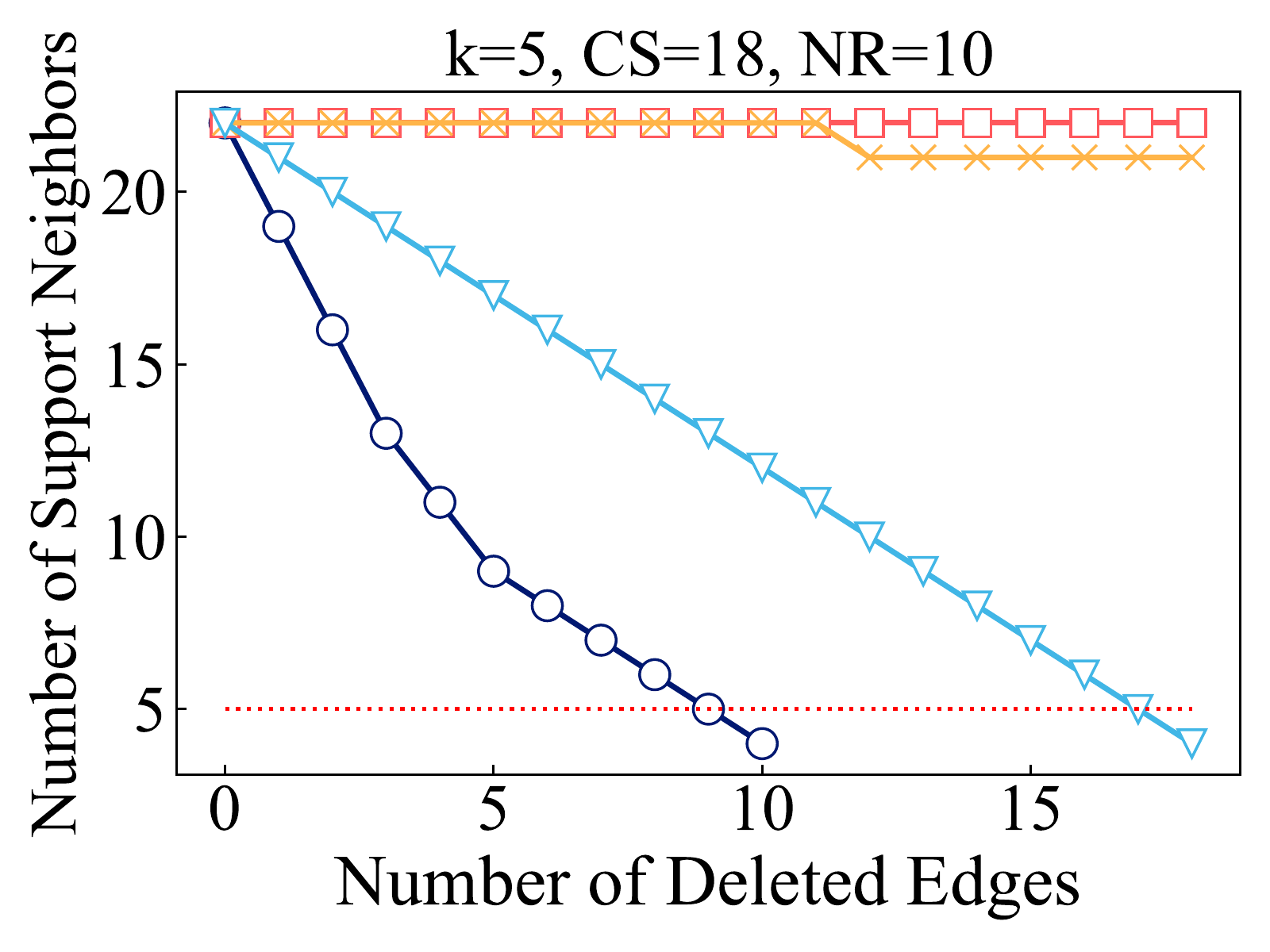}\label{fig: case_h}}\\
	\vspace{-1mm}
  \subfigure{
	\includegraphics[scale=0.29]{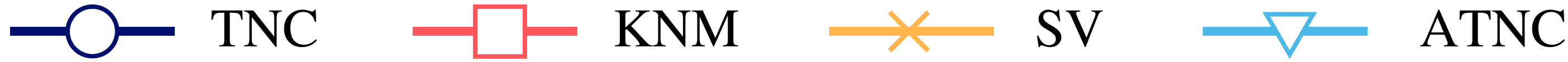}}
	\vspace{-2.5mm}
	\caption{Case study on individual nodes from $4$ mentioned networks operated by TNC, KNM, SV and ATNC. Each method is marked with a unique label and the red dotted line in each subfigure indicates the critical value of the number of supportive neighbors for current target node.}
	\label{fig: casestudy}
\end{figure*}

\textbf{Benefits of Our Proposed Methods.} Next, turning to the results generated by TNC and ATNC shown in Table \ref{tab: table2}, TNC and ATNC achieve the best and second-best performance on majority of the datasets with significant benefits over KNM and SV. For example, on CondMat network, only $259$ and $184$ bubble nodes could be detected through KNM and SV, respectively, while there are $3008$ and $2801$ bubble nodes found by TNC and ATNC, respectively, resulting in a difference of more than $10$-fold between the two sides. This definitely demonstrates that in the comparison to KNM and SV, TNC and ATNC have stronger transferability across different nodes and different networks. Besides, our proposed algorithms also perform better than KNM and SV considering SRC, WAR and RP metrics. However, we notice that the WAR of SV is a little larger than that of ATNC on RoadNet. After the observation of the bubble nodes found by SV and ATNC, there exists the situation that among the $19$ bubble nodes detected by SV, $1$ node has reduced cost of $3$ and $18$ nodes has reduced cost of $1$; while for the $4105$ bubble nodes detected by ATNC, there are $8$ nodes with reduced cost of $3$, $54$ nodes with $2$ and even $4043$ nodes with $1$. In the calculation of WAR for ATNC, the bubble nodes with reduced cost of $1$, which make up nearly $98\%$ of the total, surely have a significant diluting impact on the final result. Actually, the existence of bubble nodes with low reduced cost is common in the other networks. For instance, on Facebook network, about $65\%$ of the bubble nodes detected by ATNC have their reduced cost less than $4$ while there are $30$ nodes with reduced cost larger than $30$, and on Indo network, $90\%$ of the bubble nodes detected by ATNC have their reduced cost less than $3$ with $4$ nodes whose reduced cost larger than $10$. This is why we use weighted averaging instead of arithmetic averaging to quantify the average reduced cost of each bubble node in the network.

\textbf{Comparison between TNC and ATNC.} Reviewing what is discussed in Section \ref{sec: methodology}, it is easy to be realized that the candidates waiting to be filtered of ATNC is a subset of those of TNC. Unsurprisingly, considering the comparison between TNC and ATNC in Table \ref{tab: table2}, the performance of TNC is better than that of ATNC on the majority of networks. We also notice that on Indo and Arabic, the WAR of ATNC is slightly higher than that of TNC while the other metrics of ATNC are less than those of TNC. Taking Indo as example for analysis, we find that there are $3$ nodes with reduced cost of $8$ among the $779$ bubble detected by TNC while none of these nodes with reduced cost of $8$ explored by ATNC. This situation leads to an unfair weighting process of TNC compared to ATNC in the calculation of WAR and causes the slight difference in the final results. For the similar reason, the slight variations in the number of bubble nodes with high reduced cost lead to the difference in the final result of WAR. However, the performance of TNC is completely superior to that of ATNC on the whole. Besides, it is easy to find that TNC is not suitable for those large-scale networks, e.g., Gowalla, Citeseer, RoadNet and Google, due to the huge size of the candidates. On the contrary, ATNC is able to complete these tasks and receives appreciable results. The detailed comparisons of efficiency will be discussed in the following contents.

\textbf{Redundancy of Core Strength Metric.} As introduced before, the RP metric measures the redundancy of core strength with respect to node robustness. As mentioned in Section \ref{sec: preliminaries}, we have shown that the core strength metric does not accurately quantify the number of necessarily removed edges for making the target $k$-node collapse. From the results of ATNC in Table \ref{tab: table2}, there are more than half of the networks whose RP is larger than $10\%$ and even part of them owning RP larger than $20\%$. For example, the RP of Facebook is nearby $20\%$ and the RP of AstroPh is more than $30\%$. These results undoubtedly demonstrate that the core strength metric is not suitable for measuring the least number of edges to remove for leading the collapse to a target node.

\textbf{Efficiency of Different Methods.} The visualization for the time consumption of implementing KNM, SV, TNC and ATNC across all the mentioned networks is illustrated in Figure \ref{fig: time_all}. Overall, we can find that TNC and KNM have similar performance since they both traverse all corona nodes for edge removal in each iteration. 
Then, we can find that SV performs better than KNM and TNC on most of the networks except for HepPh network and USAir network. For HepPh network, its maximal core value $k_{max}=238$ is much higher than that of the other networks which is up to $101$. For USAir network, its size if much smaller than the others and causes the operations of SV are much more time-consuming than those of KNM and TNC.
Besides, ATNC occupies the best efficiency with significant time-consumption reduction compared to the other methods. For example, on DeezerEU network, the time consumption of SV method is about $10$ times larger than that of ARPC and the time consumption of TNC is even more than $100$ times larger than that of ATNC. And for large-scale networks, e.g., Gowalla, Citeseer, RoadNet and Google, neither TNC nor KNM can calculate the robustness for each node in those networks in the limitation of $10^5$ seconds, e.g., TNC even fails to complete the calculation of $0.1\%$ of total nodes on Google network within $10^5$ seconds, while ATNC is able to cover the task in an appreciable amount of time.

In a word, our proposed methods TNC and ATNC have significant advantages over the other baseline methods. 
And considering the much lower time complexity of ATNC compared to TNC, ATNC is more suitable to be deployed on large-scale networks for solving TNCP problem, although the effect of TNC is slightly better than that of ATNC.

  

\subsection{Case Study}
\label{sec: casestudy}

In the previous section, we provide the performance of different methods from a macroscopic perspective. Here, in this part, we offer a microscopic point of view as a case study. We visualize the variation in the number of supportive neighbors of the target node when the implementation is processing. As illustrated in Figure \ref{fig: casestudy}, $8$ individual target nodes collected from $4$ of the mentioned networks are visualized. In each subfigure, the horizontal coordinate indicates the number of removed edges during the process, the vertical coordinate indicates the number of remaining supportive neighbors of the target node after the removal. Different implemented methods are marked with different labels. Meanwhile, the red dotted line in each subfigure represents the critical number of supportive neighbors for the target node which is equal to its core value. The collapse of the target node happens when the curve drops below the red dotted line since the violation of Theorem \ref{the: coresupport}. From the examples, it is clear that fewer removed edges is needed through TNC and ATNC compared to those of KNM and SV, and TNC is able to remove fewer edges than ATNC in some cases.

\section{Application}
\label{sec: application}
Currently, $k$-core has been widely used in numerous downstream tasks, e.g., anomaly detection \cite{shin2016corescope, shin2018patterns}, community detection \cite{5991344}, detection of influential spreaders \cite{kitsak2010identification,brown2011measuring,lu2016vital,lu2016h}, etc. Laishram et al.\cite{laishram2018measuring} demonstrated that the performance of those downstream tasks is highly relative to the resilience of the $k$-core structure in a network. They proposed a heuristic metric named CIS whose calculation is based on the core strength metric. They indicated that the resilience of $k$-core is positively correlated with CIS. However, as mentioned before, we have demonstrated that the core strength metric is highly redundant for measuring the robustness of individual $k$-nodes in real-world networks. Thus, the CIS calculated from core strength, named as CS-based CIS, probably overestimates the resilience of $k$-core structures in a network. For the above reasons, we replace core strength with the node robustness achieved by ATNC algorithm in the calculation of CIS, which is named as NR-based CIS. The results of CS-based CIS and NR-based CIS on real-world networks are shown in Figure \ref{fig: application}. It is clear that on most networks, NR-based CIS is much smaller than CS-based CIS and is able to precisely measure the resilience of the $k$-core in a network. Besides, combining the information illustrated in Table \ref{tab: table2}, we can find that the difference between NR-based CIS and CS-based CIS is proportional to the RP metric, e.g., on Facebook network, ATNC provides RP=$19.59\%$ and there is a two-fold difference between CS-based CIS and NR-based CIS; while the difference on EDU network who receives RP=$0.89\%$ by ATNC is negligible. From this, it is clear that the node robustness metric has better performance, compared to the core strength metric, in precisely describing the resilience of $k$-core structures in networks.

\begin{figure}[!t]
  \centering
	\includegraphics[width=\linewidth]{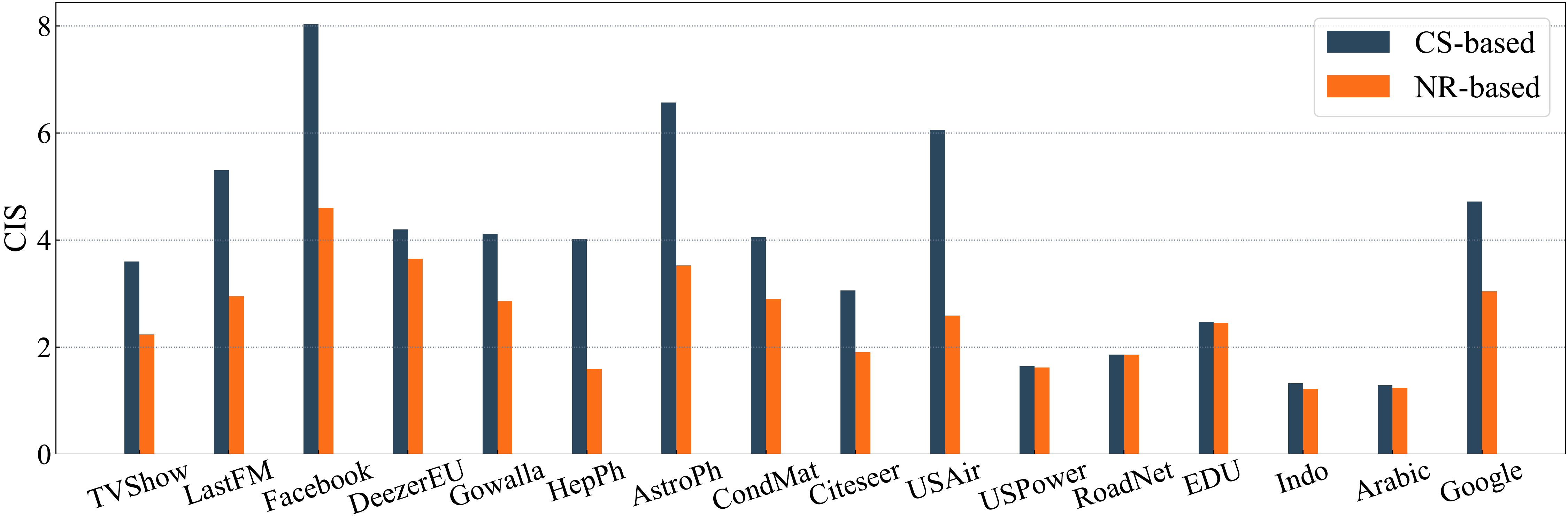}
  \vspace{-6mm}
  \caption{Comparisons of CS-based CIS and NR-based CIS on all mentioned networks. It is clear that, on most networks, NR-based CIS is able to measure the resilience of $k$-core more precisely than CS-based CIS.}
  \label{fig: application}
\end{figure}


\section{Conclusion}
\label{sec: conclusion}
In this paper, we engage in the first work on studying the robustness of individual nodes within $k$-core. We propose the TNCP problem, which aims to remove the minimal number of edges for making the target node collapse, and we also provide a proof of its NP-hardness. In order to solve TNCP problem, we propose two heuristic algorithms including TNC algorithm which exploits corona nodes to improve search efficiency, and ATNC algorithm which introduces adjacent-search strategy to further lower down computational complexity on large-scale networks. Extensive experimental results on various real-world networks, together with thorough analyses, demonstrate the superiority of our proposed methods over the baseline methods. Meanwhile, we offer the detailed processes of different algorithms being implemented on various target nodes for case study. Finally, we demonstrate that studying TNCP problem is helpful for precisely estimating the resilience of $k$-core in networks.

\begin{acks}

This work was supported in part by the Key R\&D Program of Zhejiang under Grant 2022C01018, by the National Natural Science Foundation of China under Grants 61973273 and U21B2001, by the National Key R\&D Program of China under Grant 2020YFB1006104, and by The Major Key Project of PCL under Grants PCL2022A03, PCL2021A02, and PCL2021A09.
 
\end{acks}


\bibliographystyle{ACM-Reference-Format}
\balance
\bibliography{references}

\end{document}